\documentclass[nocopyrightspace, preprint, 9pt]{sigplanconf}

\usepackage{hyperref}

\hypersetup{
    colorlinks=true,
    linkcolor=blue,
    urlcolor=blue,
    citecolor=black
}

\usepackage{tipa}
\usepackage{amsmath,amssymb,latexsym}

\usepackage{amsmath}
\usepackage{amsthm}

\usepackage{comment}
\usepackage{amsfonts}
\usepackage{amssymb}
\usepackage{mathtools}
\usepackage{commands}
\usepackage{liquidHaskell}
\usepackage{xspace}
\usepackage{epsfig}
\usepackage{booktabs}
\usepackage{listings}
\usepackage{comment}
\usepackage{stmaryrd}
\usepackage{flushend}
\usepackage{ifthen}

\newtheorem{theorem}{Theorem}
\newtheorem{lemma}{Lemma}
\newtheorem{definition}{Definition}
\newtheorem{corollary}{Corollary}

\usepackage{thmtools}
\declaretheoremstyle[%
  spaceabove=-6pt,%
  spacebelow=6pt,%
  headfont=\normalfont\itshape,%
  postheadspace=1em,%
  qed=\qedsymbol,%
  headpunct={}
]{mystyle}

\usepackage[inference]{semantic}

\usepackage{listings}
\usepackage{amssymb}

\ifdefined\withcolor
	 \definecolor{haskellblue}{rgb}{0.0, 0.0, 1.0}
	 \definecolor{haskellstr}{rgb}{0.2, 0.2, 0.6}
	 \definecolor{haskellred}{rgb}{1.0, 0.0, 0.0}
  \definecolor{gray_ulisses}{gray}{0.55}
  \definecolor{castanho_ulisses}{rgb}{0.71,0.33,0.14}
  \definecolor{preto_ulisses}{rgb}{0.41,0.20,0.04}
  \definecolor{green_ulises}{rgb}{0.2,0.75,0}
\else
	\definecolor{haskellblue}{gray}{0.1}
	\definecolor{haskellstr}{gray}{0.1}
	\definecolor{haskellred}{gray}{0.1}
	\definecolor{gray_ulisses}{gray}{0.1}
	\definecolor{castanho_ulisses}{gray}{0.1}
	\definecolor{preto_ulisses}{gray}{0.1}
	\definecolor{green_ulisses}{gray}{0.1}
\fi

\definecolor{lcolor}{gray}{0.0}
\definecolor{lappcolor}{gray}{0.0}
\definecolor{lappascolor}{gray}{0.0}

\def\codesize{\normalsize}

\newcommand\showfocus[1]{\color{purple}{\textbf{#1}}}

\lstdefinelanguage{HaskellUlisses} {
	basicstyle=\ttfamily\codesize,
	moredelim=[is][\showfocus]{\#}{\#},
	sensitive=true,
	morecomment=[l][\color{gray_ulisses}\ttfamily\itshape\codesize]{--},
	morecomment=[s][\color{gray_ulisses}\ttfamily\itshape\codesize]{\{-}{-\}},
	morestring=[b]",
	stringstyle=\color{haskellstr},
	showstringspaces=false,
	numberstyle=\codesize,
	numberblanklines=true,
	showspaces=false,
	breaklines=true,
	showtabs=false,
  literate={
	         {<!}{{{\color{lcolor}<!}}}2
           {==!}{{{\color{lcolor}==!}}}3
           {`}{{{$^{\backprime}{}$}}}1
           {?}{{{\color{lcolor}?}}}1
           {<=}{{$\leq$}}1
           {theta}{{$\theta$}}1
           {rf}{{{\color{lappcolor}f}}}1
           {gf}{{{\color{lappascolor}f}}}1
           {rmap}{{{\color{lappcolor}map}}}3
           {gmap}{{{\color{lappascolor}map}}}3
           {r.}{{{\color{lappcolor}.}}}1
           {g.}{{{\color{lappascolor}.}}}1
           {r++}{{{\showfocus{++}}}}2
           {g++}{{{\color{lappascolor}++}}}2
           {gfib}{{{\color{lappascolor}fib}}}3
           {rfib}{{{\color{lappcolor}fib}}}3
           {r++}{{{\color{lappcolor}++}}}2
           {op}{{$\odot$}}1
           {env}{{$\Gamma$}}1
           {|-}{{$\vdash$}}1
           {<=!}{{{\color{lcolor}<=!}}}3
           {!=}{{$\neq$}}1
           {forall}{{$\forall$}}1
           {->}{{$\rightarrow$}}2
           {=*}{{$\eqfun$}}2
           {Set_mem}{{$\in$}}1
           {Set_cup}{{$\cup$}}1
           {Set_cap}{{$\cap$}}1
           {Set_emp}{{$\emptyset$}}1
           {Set_sub}{{$\subseteq$}}1
           {<=>}{{$\Leftrightarrow$}}3
           {=>}{{$\Rightarrow$}}2
           {||-}{{$\vdash$}}1
           {|->}{{$\mapsto$}}2
           {<:}{{$\preceq$}}1
           {Inarritu}{Inarritu}8},
	emph=
	{[1]
		FilePath,IOError,abs,acos,acosh,all,and,any,appendFile,approxRational,asTypeOf,asin,
		asinh,atan,atan2,atanh,basicIORun,break,catch,ceiling,chr,compare,concat,concatMap,
		const,cos,cosh,curry,cycle,decodeFloat,denominator,digitToInt,div,divMod,drop,
		dropWhile,either,elem,encodeFloat,enumFrom,enumFromThen,enumFromThenTo,enumFromTo,
		error,even,exp,exponent,fail,mapMaybe,filter,flip,floatDigits,floatRadix,floatRange,floor,
		fmap,foldl,foldl1,foldr,foldr1,fromDouble,fromEnum,fromInt,fromInteger,fromIntegral,
		fromRational,fst,gcd,getChar,getContents,getLine,head,id,inRange,index,init,intToDigit,
		interact,ioError,isAlpha,isAlphaNum,isAscii,isControl,isDenormalized,isDigit,isHexDigit,
		isIEEE,isInfinite,isLower,isNaN,isNegativeZero,isOctDigit,isPrint,isSpace,isUpper,iterate,
		last,lcm,length,lex,lexDigits,lexLitChar,lines,log,logBase,lookup,map,mapM,mapM_,max,
		maxBound,posMax,negMax,maximum,maybe,min,minBound,minimum,mod,negate,not,notElem,null,numerator,odd,
		or,ord,pi,pred,primExitWith,print,product,properFraction,putChar,putStr,putStrLn,quot,
		quotRem,range,rangeSize,read,readDec,readFile,readFloat,readHex,readIO,readInt,readList,readLitChar,
		readLn,readOct,readParen,readSigned,reads,readsPrec,realToFrac,recip,rem,repeat,replicate,return,
		reverse,round,scaleFloat,scanl,scanl1,scanr,scanr1,seq,sequence,sequence_,show,showChar,showInt,
		showList,showLitChar,showParen,showSigned,showString,shows,showsPrec,significand,signum,sin,
		sinh,snd,span,splitAt,sqrt,subtract,succ,sum,tail,take,takeWhile,tan,tanh,threadToIOResult,toEnum,
		toInt,toInteger,toLower,toRational,toUpper,truncate,uncurry,undefined,unlines,until,unwords,unzip,
		unzip3,userError,words,writeFile,zip,zip3,zipWith,zipWith3,listArray,doParse,empty,for,initTo,
        assert,compose,checkGE,maxEvens,empty,create,get,set,initialize,idVec,fastFib,fibMemo,
        ex1,ex2,ex3,incr,inc,dec,isPos,positives,find,insert,len,size,union,fromList,initUpto,trim,
        insertSort,decsort,qsort,reverse,append,upperCase, ifM, whileM, get, decrM, diff,
        project, select, leq, elts, keys, dkeys, dfun, addKey, pTrue, emptyRD, rFalse,
        	dom, rng, isI, isD, isS, movie1, movie2,  toI, toS, toD, good_titles, runState, ret,
        	update, getCtr, setCtr, ctr, rdCtr, wrCtr, ifTest, whileTest, posCtr, zeroCtr, decr, decCtr,
        	pread , pwrite , plookup , pcontents, pcreateF , pcreateFP, pcreateD, active, caps, pset, eqP,
        	write, contents, alloc, derivP, copyP, createDir, store, copyRec, copySpec,
        	forM_, when, flookup, fread, createDir, pcreateFile, isFile, copyFrame, ?
	},
	emphstyle={[1]\color{haskellblue}},
	emph=
	{[2] 	Show,Eq,Ord,Num,UpClosed,Comp,Wit,Witness,Inductive,Meet,Flip,TRUE,Nat,Pos,Neg,IntGE,Plus,List,
        Bool,Char,Double,Either,Float,IO,Integer,Int,Maybe,
        Ordering,Rational,Ratio,ReadS,ShowS,String,Word8,
        InPacket,Tree,Vec,NullTerm,IncrList,DecrList,
        UniqList,BST,MinHeap,MaxHeap,World,RIO,IO,HIO,Post,Pre,
        Privilege, Prop, Chain, ChainTy, Range, Dict, RD, Dom, Set, P, Univ, Schema, MovieSchema, RT,
        TDom, TRange, MoviesTable, RTSubEqFlds, RTEqFlds, Disjoint, Union, Ret, Seq, Trans, Map,
        Pure, Then, Else, Exit, Inv, OneState, Priv, Path, FH, Stable,
				Prop, Nat,
	},
	emphstyle={[2]\color{castanho_ulisses}},
	emph=
	{[3]
		case,class,data,deriving,do,else,if,import,in,infixl,infixr,instance,let,
		module,of,primitive,then,refinement,type,where,forall,bound,
		measure,reflect,predicate
	},
	emphstyle={[3]\color{preto_ulisses}\textbf},
	emph=
	{[4]
		quot,rem,div,mod,elem,notElem,seq
	},
	emphstyle={[4]\color{castanho_ulisses}\textbf},
	emph=
	{[5]
		EQ,GT,LT,Left,Right
	},
	emphstyle={[5]\color{preto_ulisses}\textbf},
	emph=
	{[6]
	    axiomatize, measure, inline
	},
	emphstyle={[6]\color{lcolor}}
}

\lstnewenvironment{code}
{\lstset{language=HaskellUlisses}}
{}

\lstnewenvironment{mcode}
{\lstset{language=HaskellUlisses,columns=fullflexible,keepspaces,mathescape}}
{}

\lstMakeShortInline[language=HaskellUlisses,mathescape,keepspaces,mathescape,basicstyle=\ttfamily\normalsize,breakatwhitespace]@

\lstdefinelanguage{Pseudo} {
	basicstyle=\ttfamily\codesize,
	sensitive=true,
  mathescape=true,
	morecomment=[l][\color{gray_ulisses}\ttfamily\codesize]{--},
	morecomment=[s][\color{gray_ulisses}\ttfamily\codesize]{\{-}{-\}},
	morestring=[b]",
	showstringspaces=false,
	numberstyle=\codesize,
	numberblanklines=true,
	showspaces=false,
	breaklines=true,
	showtabs=false
}

\newcommand{\isTechReport}{false} % true or false

\begin{document}

%\title{From Types to Proofs via Refinement Reflection}
\title{\ifthenelse{\equal{\isTechReport}{true}}
       {{Technical Report:}}
        {}
        Refinement Reflection
       }
\subtitle{\emph{(or, how to turn your favorite language into a proof assistant using SMT)}}

% Utopia: Liquid Types as a Proof Assistant
% HeLP:   Haskell as a Liquid Type Proof Assistant
%	REPAIr:	REfinement Proof AssIstant

% \authorinfo{OMITTED FOR BLIND REVIEW}{}
\authorinfo{Niki Vazou \and Ranjit Jhala}{University of California, San Diego}{}

\maketitle

\begin{abstract}
Refinement Reflection turns your favorite programming 
language into a proof assistant by reflecting
the ​code implementing a​ user-defined function 
into the function's (output) refinement type. 
As a consequence, at \emph{uses} of the function, 
the function definition is unfolded into the refinement logic 
in a precise, predictable and most
importantly, programmer controllable way.
In the logic, we encode functions and lambdas using uninterpreted symbols
preserving SMT-based decidable verification. 
In the language, we provide a library of combinators that lets programmers 
compose proofs from basic refinements and function definitions.
We have implemented our approach in the Liquid Haskell system, 
thereby converting Haskell into an interactive proof assistant, 
that we used to verify a variety of properties ranging 
from arithmetic properties of higher order, recursive functions
to the Monoid, Applicative, Functor and Monad type class laws 
for a variety of instances.
\end{abstract}
\section{Introduction}\label{sec:intro}

Wouldn't it be great to write proofs \emph{of}
programs in your favorite language, \emph{by}
writing programs in your favorite language,
allowing you to avail of verification, while
reusing the libraries, compilers and run-times
\emph{for} your favorite language?

Refinement types \citep{ConstableS87,Rushby98} offer a
form of programming with proofs that can be
retrofitted into several languages like
ML~\cite{pfenningxi98,GordonRefinement09,LiquidPLDI08},
C~\cite{deputy,LiquidPOPL10},
Haskell~\citep{Vazou14}, TypeScript~\cite{Vekris16}
and Racket~\cite{RefinedRacket}.
The retrofitting relies upon restricting refinements
to so-called ``shallow'' specifications that
correspond to \emph{abstract} interpretations
of the behavior of functions.
For example, refinements make it easy to specify
that the list returned by the @append@ function
has size equal to the sum of those of its inputs.
These shallow specifications fall within decidable
logical fragments, and hence, can be automatically
verified using SMT based refinement typing.

Refinements are a pale shadow of what is possible
with dependently typed languages like Coq, Agda
and Idris which permit ``deep'' specification
and verification.
These languages come equipped with mechanisms
that \emph{represent} and \emph{manipulate} the
exact descriptions of user-defined functions.
For example, we can represent the specification
that the @append@ function is associative, and we
can manipulate (unfold) its definition to write a
small program that constructs a proof of the
specification.
Dafny~\citep{dafny}, \fstar~\citep{fstar} and
Halo~\citep{HALO} take a step towards
SMT-based deep verification, by encoding
user-defined functions as universally
quantified logical formulas or ``axioms''.
This axiomatic approach offers significant automation
but is a devil's bargain as by relying heavily upon
brittle heuristics for ``triggering'' axiom instantiation,
it gives away decidable, and hence, predictable
verification~\citep{Leino16}.
%
% In practice, the heuristics may fail to apply or may
% go into infinite ``matching loops'' leaving programmers
% with little clue of how to proceed~\citep{Leino16}.
%
% Indeed, according to the authors of Dafny, these heuristics
% are unpredictable and suffer from a ``butterfly effect''
% where modification in a part of the program changes the
% outcome of verification in other part \cite{Leino16}.

\mypara{Refinement Reflection}
In this paper, we present a new approach to retrofitting
deep verification into existing languages. Our approach
reconciles the automation of SMT-based refinement typing
with decidable and predictable verification, and enables
users to reify pencil-and-paper proofs simply
as programs in the source language.
Our key insight is dead simple: the ​code
implementing a​ user-defined function can
be \emph{reflected}​ into the function's
(output) refinement type, thus converting
the function's (refinement) type signature
into a deep specification of the functions
behavior.
This simple idea has a profound consequence:
at \emph{uses} of the function, the standard
rule for (dependent) function application
yields a precise, predictable and most
importantly, programmer controllable
means of \emph{instantiating} the deep
specification that is not tethered to
brittle SMT heuristics.
Reflection captures deep specifications as
refinements, but poses challenges for the
\emph{logic} and \emph{language}.

\mypara{Logic: Algorithmic Verification}
Our first challenge: how can we \emph{encode terms}
from an expressive higher order language in a decidable
refinement logic in order to retain decidable, and hence,
predictable, verification?
We address this problem by using ideas for defunctionalization
from the theorem proving literature which
encode functions and lambdas using uninterpreted symbols.
This encoding lets us use (SMT-based) congruence closure to
reason about equality~(\S~\ref{sec:algorithmic}).
Of course, congruence is not enough; in general, \eg to prove
two functions extensionally equal, we require facilities for
manipulating function definitions.

\mypara{Language: Proof Composition}
Thus, as we wish to retrofit proofs into
existing languages, our second challenge:
how can we design a \emph{library of combinators}
that lets programmers \emph{compose proofs}
from basic refinements and function definitions?
We develop such a library, wherein proofs
are represented simply as unit-values
refined by the proposition that they
are proofs of. % ~(\S~\ref{sec:library}).
Refinement reflection lets us unfold definitions
simply by \emph{applying} the function to the
relevant inputs, and finally, we show how to
build up sophisticated proofs using a small
library of combinators that permit reasoning
in an algebraic or equational style.

\mypara{Implementation \& Evaluation}
We have implemented our approach in the \toolname
system, thereby retrofitting deep verification into
Haskell, converting it into an interactive proof
assistant.
\toolname's refinement types crucially allow us to
soundly account for the dreaded bottom by checking
that (refined) functions produce (non-bottom)
values~\cite{Vazou14}.
We evaluate our approach by using \toolname to
verify a variety of properties including arithmetic
properties of higher order, recursive functions,
textbook theorems about functions on inductively
defined datatypes, and the Monoid, Applicative,
Functor and Monad type class laws for a variety
of instances. %(\S~\ref{sec:evaluation}).
We demonstrate that our proofs look very much like
transcriptions of their pencil-and-paper analogues.
Yet, the proofs are plain Haskell functions, where
case-splitting and induction are performed by
plain pattern-matching and recursion.

% showing, that it may be possible to avail of the benefits
% of deep specification and verification without
% leaving the comforts of your favorite programming
% language.

To summarize, this paper describes a means of
retrofitting deep specification and verification
into your favorite language, by making the
following contributions:
\begin{itemize}
\item We start with an informal description of
      refinement reflection, and how it can
      be used to prove theorems about functions,
      by writing functions~(\S~\ref{sec:overview}).

\item We formalize refinement reflection using
      a core calculus, and prove it sound with
      respect to a denotational
      semantics~(\S~\ref{sec:formalism}).

\item We show how to keep type checking
      decidable~(\S~\ref{sec:algorithmic})
      while using uninterpreted functions and
      defunctionalization to reason about
      extensional equality in higher-order
      specifications~(\S~\ref{sec:lambdas}).

\item Finally, we have implemented refinement reflection
      in \toolname, a refinement type system for Haskell.
      We develop a library of (refined) proof combinators
      and evaluate our approach by proving various theorems
      about recursive, higher-order functions operating over
      integers and algebraic data types~(\S~\ref{sec:evaluation}).
\end{itemize}

\section{Overview}
\label{sec:overview}
\label{sec:examples}

We begin with a fast overview of refinement reflection and
how it allows us to write proofs \emph{of} and \emph{by}
Haskell functions.

\subsection{Refinement Types}

First, we recall some preliminaries about refinement types
and how they enable shallow specification and verification.

\mypara{Refinement types} are the source program's (here
Haskell's) types decorated with logical predicates drawn
from a(n SMT decidable) logic~\citep{ConstableS87,Rushby98}.
For example, we can refine Haskell's @Int@ datatype with a
predicate @0 <= v@, to get a @Nat@ type:
\begin{code}
  type Nat = {v:Int | 0 <= v}
\end{code}
The variable @v@ names the value described by the type,
hence the above can be read as the
``set of @Int@ values @v@ that are greater than 0".
The refinement is drawn from the logic of quantifier
free linear arithmetic and uninterpreted functions
(QF-UFLIA~\cite{SMTLIB2}).

\mypara{Specification \& Verification}
We can use refinements to define and type the
textbook Fibonacci function as:
\begin{code}
  fib   :: Nat -> Nat
  fib 0 = 0
  fib 1 = 1
  fib n = fib (n-1) + fib (n-2)
\end{code}
Here, the input type's refinement specifies a
\emph{pre-condition} that the parameters must
be @Nat@, which is needed to ensure termination,
and the output types's refinement specifies a
\emph{post-condition} that the result is also a @Nat@.
Thus refinement type checking, lets us specify
and (automatically) verify the shallow property
that if @fib@ is invoked with non-negative
@Int@ values, then it (terminates) and yields
a non-negative value.

\mypara{Propositions}
We can use refinements to define a data type
representing propositions simply as an alias
for unit, a data type that carries no run-time
information:
\begin{mcode}
  type $\typp$ = ()
\end{mcode}
but which can be \emph{refined} with desired
propositions about the code.
For example, the following states the proposition
$2 + 2$ equals $4$.
\begin{mcode}
  type Plus_2_2_eq_4 = {v: $\typp$ | 2 + 2 = 4}
\end{mcode}
For clarity, we abbreviate the above type by omitting
the irrelevant basic type $\typp$ and variable @v@:
\begin{mcode}
  type Plus_2_2_eq_4 = {2 + 2 = 4}
\end{mcode}
We represent universally quantified propositions as function types:
\begin{mcode}
  type Plus_com = x:Int->y:Int->{x+y = y+x}
\end{mcode}
Here, the parameters @x@ and @y@ refer to input
values; any inhabitant of the above type is a
proof that @Int@ addition is commutative.

\mypara{Proofs}
We can now \emph{prove} the above theorems simply by
writing Haskell programs. To ease this task \toolname
provides primitives to construct proof terms by
``casting'' expressions to \typp.
\begin{mcode}
  data QED = QED

  (**) :: a -> QED -> $\typp$
  _ ** _  = ()
\end{mcode}
To resemble mathematical proofs, we make this casting post-fix.
Thus, we can write @e ** QED@ to cast @e@ to a value of \typp.
For example, we can prove the above propositions simply by writing
\begin{code}
  pf_plus_2_2 :: Plus_2_2_eq_4
  pf_plus_2_2 = trivial ** QED

  pf_plus_comm :: Plus_comm
  pf_plus_comm = \x y -> trivial ** QED

  trivial = ()
\end{code}
Via standard refinement type checking, the above code yields
the respective verification conditions (VCs),
\begin{align*}
                   2 + 2 & = 4 \\
\forall \ x\ y\ .\ x + y & = y + x
\end{align*}
which are easily proved valid by the SMT solver, allowing us
to prove the respective propositions.

\mypara{A Note on Bottom:} Readers familiar with Haskell's
semantics may be feeling a bit anxious about whether the
dreaded ``bottom", which inhabits all types, makes our
proofs suspect.
Fortunately, as described in \cite{Vazou14}, \toolname
ensures that all terms with non-trivial refinements
provably evaluate to (non-bottom) values, thereby making
our proofs sound.

\subsection{Refinement Reflection}

Suppose that we wish to prove properties about the @fib@
function, \eg that @fib 2@ equals @1@.
\begin{code}
  type fib2_eq_1 = { fib 2 = 1 }
\end{code}
%
%% \NV{By Standard refinement type checking, you mean liquid types, not FStar}
Standard refinement type checking runs into two problems.
First, for decidability and soundness, arbitrary user-defined
functions do not belong the refinement logic, \ie we cannot even
\emph{refer} to @fib@ in a refinement.
Second, the only information that a refinement type checker
has about the behavior of @fib@ is its shallow type
specification @Nat -> Nat@ which is far too weak to verify
@fib2_eq_1@.
To address both problems, we use the following annotation,
which sets in motion the three steps of refinement reflection:
\begin{code}
  reflect fib
\end{code}

\mypara{Step 1: Definition}
The annotation tells \toolname to declare an
\emph{uninterpreted function} @fib :: Int -> Int@
in the refinement logic.
By uninterpreted, we mean that the logical @fib@
is \emph{not} connected to the program function
@fib@; as far as the logic is concerned, @fib@
only satisfies the \emph{congruence axiom}
$$\forall n, m.\ n = m\ \Rightarrow\ \fib{n} = \fib{m}$$
On its own, the uninterpreted function is not
terribly useful, as it does not let us prove
% It lets us prove theorems like
% $$\forall m,\ n.\ m = n \Rightarrow \fib{m} = \fib{n}$$
%
%% \begin{code}
  %% fib_cong :: n:Nat -> m:Nat -> {m=n => fib m = fib n}
  %% fib_cong = trivial ** QED
%% \end{code}
%% %
%but not
@fib2_eq_1@ which requires reasoning about the
\emph{definition} of @fib@.

\mypara{Step 2: Reflection}
In the next key step, \toolname reflects the
definition into the refinement type of @fib@
by automatically strengthening the user defined
type for @fib@ to:
\begin{code}
  fib :: n:Nat -> {v:Nat | fibP v n}
\end{code}
where @fibP@ is an alias for a refinement
\emph{automatically derived} from the
function's definition:
\begin{mcode}
  predicate fibP v n =
    v = if n = 0 then 0 else
        if n = 1 then 1 else
        fib(n-1) + fib(n-2)
\end{mcode}

\mypara{Step 3: Application}
With the reflected refinement type,
each application of @fib@ in the code
automatically unfolds the @fib@ definition
\textit{once} during refinement type checking.
We prove @fib2_eq_1@ by:
\begin{code}
  pf_fib2 :: { fib 2 = 1 }
  pf_fib2 = #fib# 2 == #fib# 1 + #fib# 0 ** QED
\end{code}
We write @#f#@ to denote places where the
unfolding of @f@'s definition is important.
The proof is verified as the above is A-normalized to
\begin{code}
  let { t0 = fib 0; t1 = fib 1; t2 = fib 2 }
  in  ( t2 == t1 + t0 ) ** QED
\end{code}
Which via standard refinement typing, yields the
following verification condition that is easily
discharged by the SMT solver, even though @fib@
is uninterpreted:
\begin{align*}
   & (\fibdef\ (\fib\ 0)\ 0) \ \wedge\ (\fibdef\ (\fib\ 1)\ 1) \ \wedge\ (\fibdef\ (\fib\ 2)\ 2) \\
   & \Rightarrow\ (\fib{2} = 1)
\end{align*}
Note that the verification of @pf_fib2@ relies
merely on the fact that @fib@ was applied
to (\ie unfolded at) @0@, @1@ and @2@.
The SMT solver can automatically \emph{combine}
the facts, once they are in the antecedent.
Hence, the following would also be verified:
\begin{code}
  pf_fib2' :: { fib 2 = 1 }
  pf_fib2' = [#fib# 0, #fib# 1, #fib# 2] ** QED
\end{code}
Thus, unlike classical dependent typing, refinement
reflection \emph{does not} perform any type-level
computation.

\mypara{Reflection vs. Axiomatization}
An alternative \emph{axiomatic} approach, used by Dafny,
\fstar and HALO, is to encode the definition of @fib@ as
a universally quantified SMT formula (or axiom):
$$\forall n.\ \fibdef\ (\fib\ n)\ n$$
Axiomatization offers greater automation than
reflection. Unlike \toolname, Dafny
%and \fstar
will verify the equivalent of the following by
\emph{automatically instantiating} the above
axiom at @2@, @1@ and @0@:
\begin{code}
  axPf_fib2 :: { fib 2 = 1 }
  axPf_fib2 = trivial ** QED
\end{code}

However, the presence of such axioms renders checking
the VCs undecidable. In practice, automatic axiom
instantation can easily lead to infinite ``matching loops''.
For example, the existence of a term \fib{n} in a VC
can trigger the above axiom, which may then produce
the terms \fib{(n-1)} and \fib{(n-2)}, which may then
recursively give rise to further instantiations
\emph{ad infinitum}.
To prevent matching loops an expert must carefully
craft ``triggers'' and provide a ``fuel''
parameter~\citep{Amin2014ComputingWA} that can be
used to restrict the numbers of the SMT unfoldings,
which ensure termination, but can cause the axiom
to not be instantiated at the right places.
In short, the undecidability of the VC checking
and its attendant heuristics makes verification
unpredictable~\citep{Leino16}.

\subsection{Structuring Proofs}

In contrast to the axiomatic approach,
with refinement reflection, the VCs are
deliberately designed to always fall in
an SMT-decidable logic, as function symbols
are uninterpreted.
It is upto the programmer to unfold the
definitions at the appropriate places,
which we have found, with careful design
of proof combinators, to be quite
a natural and pleasant experience.
To this end, we have developed a library
of proof combinators that permits reasoning
about equalities and linear arithmetic,
inspired by Agda~\citep{agdaequational}.

\mypara{``Equation'' Combinators}
We equip \toolname with a a family of
equation combinators @op.@ for each
logical operator @op@ in
$\{=, \not =, \leq, <, \geq, > \}$,
the operators in the theory QF-UFLIA.
The refinement type of @op.@  \emph{requires}
that $x \odot y$ holds and then \emph{ensures}
that the returned value is equal to @x@.
For example, we define @=.@ as:
\begin{code}
  (=.) :: x:a -> y:{a| x=y} -> {v:a| v=x}
  x =. _ = x
\end{code}
and use it to write the following ``equational" proof:
\begin{code}
  eqPf_fib2 :: { fib 2 = 1 }
  eqPf_fib2 =  #fib# 2
            =. #fib# 1 + #fib# 0
            =. 1
            ** QED
\end{code} %$

\mypara{``Because'' Combinators}
Often, we need to compose ``lemmas'' into larger
theorems. For example, to prove @fib 3 = 2@ we
may wish to reuse @eqPf_fib2@ as a lemma.
To this end, \toolname has a ``because'' combinator:
\begin{mcode}
  ($\because$) :: ($\typp$ -> a) -> $\typp$ -> a
  f $\because$ y = f y
\end{mcode}
The operator is simply an alias for function
application that lets us write
@ x op. y $\because$ p@ (instead of @(op.) x y p@)
where @(op.)@ is extended to accept an \textit{optional} third proof
argument via Haskell's type class mechanisms.
We can use the because combinator to
prove that @fib 3 = 2@ just by writing
plain Haskell code:
\begin{mcode}
  eqPf_fib3 :: {fib 3 = 2}
  eqPf_fib3 =  #fib# 3
            =. fib 2 + #fib# 1
            =. 2              $\because$ eqPf_fib2
            ** QED
\end{mcode}

\mypara{Arithmetic and Ordering}
SMT based refinements let us go well beyond just equational
reasoning. Next, lets see how we can use arithmetic and
ordering to prove that @fib@ is (locally) increasing,
\ie for all $n$, $\fib{n} \leq \fib{(n+1)}$
\begin{mcode}
  fibUp :: n:Nat -> {fib n <= fib (n+1)}
  fibUp n
    | n == 0
    =  #fib# 0 <. #fib# 1
    ** QED

    | n == 1
    =  fib 1 <=. fib 1 + fib 0 <=. #fib# 2
    ** QED

    | otherwise
    =  #fib# n
    =. fib (n-1) + fib (n-2)
    <=. fib n     + fib (n-2) $\because$ fibUp (n-1)
    <=. fib n     + fib (n-1) $\because$ fibUp (n-2)
    <=. #fib# (n+1)
    ** QED
\end{mcode} %$

\mypara{Case Splitting and Induction}
The proof @fibUp@ works by induction on @n@.
In the \emph{base} cases @0@ and @1@, we simply assert
the relevant inequalities. These are verified as the
reflected refinement unfolds the definition of
@fib@ at those inputs.
The derived VCs are (automatically) proved
as the SMT solver concludes $0 < 1$ and $1 + 0 \leq 1$
respectively.
In the \emph{inductive} case, @fib n@ is unfolded
to  @fib (n-1) + fib (n-2)@, which, because of the
induction hypothesis (applied by invoking @fibUp@
at @n-1@ and @n-2@), and the SMT solvers arithmetic
reasoning, completes the proof.

\mypara{Higher Order Theorems}
Refinements smoothly accomodate higher-order reasoning.
For example, lets prove that every locally increasing
function is monotonic, \ie
if @f z <= f (z+1)@ for all @z@,
then @f x <= f y@ for all @x < y@.
\begin{mcode}
  fMono :: f:(Nat -> Int)
        -> fUp:(z:Nat -> {f z <= f (z+1)})
        -> x:Nat
        -> y:{x < y}
        -> {f x <= f y} / [y]
  fMono f inc x y
    | x + 1 == y
    =  f x <=. f (x+1) $\because$ fUp x
           <=. f y
           ** QED

    | x + 1 < y
    =  f x <=. f (y-1) $\because$ fMono f fUp x (y-1)
           <=. f y     $\because$ fUp (y-1)
           ** QED
\end{mcode}
We prove the theorem by induction
on @y@, which is specified by the
annotation @/ [y]@ which states
that @y@ is a well-founded
termination metric that decreases
at each recursive call~\citep{Vazou14}.
%
% All reflected functions are proved terminating.
% When the annotation metric is not explicit Liquid Haskell
% successfully uses heuristics to automatically prove termination. 
%
If @x+1 == y@, then we use @fUp x@.
Otherwise, @x+1 < y@, and we use
the induction hypothesis \ie apply
@fMono@ at @y-1@, after which
transitivity of the less-than
ordering finishes the proof.
We can use the general @fMono@
theorem to prove that @fib@
increases monotonically:
\begin{code}
  fibMono :: n:Nat -> m:{n<m} ->
             {fib n <= fib m}
  fibMono = fMono fib fibUp
\end{code}

\subsection{Case Study: Peano Numerals}

Refinement reflection is not limited to programs
operating on integers. We conclude the overview
with a small library for Peano numerals, defined
via the following algebraic data type:
\begin{code}
  data Peano = Z | S Peano
\end{code}
We can @add@ two @Peano@ numbers via:
\begin{code}
  reflect add :: Peano -> Peano -> Peano
  add Z     m = m
  add (S n) m = S (add n m)
\end{code}
In \S~\ref{subsec:measures} we will describe
exactly how the reflection mechanism (illustrated
via @fibP@) is extended to account for ADTs like @Peano@.
Note that \toolname automatically checks
that @add@ is total~\citep{Vazou14}, which
lets us safely @reflect@ it into the
refinement logic.

\mypara{Add Zero to Left}
As an easy warm up, lets show that
adding zero to the left leaves the
number unchanged:
\begin{code}
  zeroL :: n:Peano -> { add Z n == n }
  zeroL n =  #add# Z n
          =. n
          ** QED
\end{code}

\mypara{Add Zero to Right}
It is slightly more work to prove
that adding zero to the right also
leaves the number unchanged.
\begin{mcode}
  zeroR :: n:Peano -> { add n Z == n }
  zeroR Z     =  #add# Z Z
              =. Z
              ** QED

  zeroR (S n) =  #add# (S n) Z
              =. S (add n Z)
              =. S n         $\because$     zeroR n
              ** QED
\end{mcode}
The proof goes by induction, splitting cases on
whether the number is zero or non-zero. Consequently,
we pattern match on the parameter @n@, and furnish
separate proofs for each case.
In the ``zero" case, we simply unfold the definition
of @add@.
In the ``successor" case, after unfolding we (literally)
apply the induction hypothesis by using the because operator.
\toolname's termination and totality checker
verifies that we are in fact doing induction
properly, \ie the recursion in @zeroR@ is
well-founded~(\S~\ref{sec:types-reflection}).

\mypara{Commutativity}
Lets conclude by proving that @add@ is commutative:
\begin{mcode}
add_com :: a:_ -> b:_ -> {add a b = add b a}

add_com    Z  b =  #add# Z b
                =. b
                =. add b Z     $\because$ zeroR b
                ** QED

add_com (S a) b =  #add# (S a) b
                =. S (add a b)
                =. S (add b a) $\because$ add_com a b
                =. add b (S a) $\because$ sucR b a
                ** QED
\end{mcode}
using a lemma @sucR@
\begin{code}
  sucR :: n:Peano -> m:Peano ->
                {add n (S m) = S (add n m)}
  sucR = exercise_for_reader
\end{code}

%
% Again, note how case-splitting, induction, and
% lemma reuse are just pattern-matching, recursion and
% function application, via the lemma @sucR@:

Thus, refinement reflection lets us prove properties of Haskell programs
just by writing Haskell programs: lemmas are just functions, case-splitting
is just branching and pattern matching, and induction is just recursion.
Next, we formalize refinement reflection and describe
how to keep type checking decidable and predictable.

\section{Refinement Reflection}
\label{sec:formalism}
\label{sec:types-reflection}

Our first step towards formalizing refinement
reflection is a core calculus \corelan with an
\emph{undecidable} type system based on
denotational semantics.
We show how the soundness of the type system
allows us to \emph{prove theorems} using \corelan.

%
%% Note that \smtlan programs are a subset of
%% \corelan programs derivations
%%
%% a subset of \corelan where the that forms a
%% decidable logic of
%% refinements, and use it to obtain \corelan with
%% decidable SMT-based algorithmic typing.

\subsection{Syntax}
\newcommand{\ra}[1]{\renewcommand{\arraystretch}{#1}}

\begin{figure}[t!]
\centering
$$
\begin{array}{rrcl}
\emphbf{Operators} \quad
  & \odot
  & ::= & = \spmid  <
\\[0.03in]

\emphbf{Constants} \quad
  & c
  & ::=
  & \land \spmid \lnot \spmid \odot \spmid +,-,\dots  \\
  && \spmid & \etrue\spmid \efalse \spmid 0, 1,-1, \dots
\\[0.03in]

\emphbf{Values} \quad
  & w & ::=&  c
             \spmid \efun{x}{\typ}{e} \spmid D\ \overline{w}
\\[0.03in]

\emphbf{Expressions} \quad
  & e & ::=    & w \spmid x \spmid \eapp{e}{e}      \\
  &   & \spmid & \ecase{x}{e}{\dc}{\overline{x}}{e}
  %% &   & \spmid & \eletb{x}{\gtyp}{e}{e}
\\[0.03in]

\emphbf{Binders} \quad
  & \bd & ::= & e \spmid \eletb{x}{\gtyp}{\bd}{\bd}
\\[0.03in]

\emphbf{Program} \quad
  & \prog & ::= & \bd \spmid \erefb{x}{\gtyp}{e}{\prog}
\\[0.03in]

% \emphbf{Logical Labels} \quad
%  & L & ::= & \la \spmid \lm
%\\[0.03in]

\emphbf{Basic Types} \quad
  & \btyp
  & ::=
  & \tint \spmid \tbool \spmid T
\\[0.03in]

\emphbf{Refined Types} \quad
  & \typ
  & ::=      & \tref{v}{\btyp}{\reft} \spmid \tfun{x}{\typ}{\typ}
\\[0.05in]
\end{array}
$$
\caption{\textbf{Syntax of \corelan}}
\label{fig:syntax}
\end{figure}

\begin{figure}[t!]
\centering

\emphbf{Contexts}\hfill{{$\fbox{\textit{C}}$}}
$$
\begin{array}{rcl}
  C & ::=    & \bullet \\
    & \spmid & C\ e \spmid c\ C \spmid D\ \overline{e}\ C\ \overline{e}\\
    & \spmid & \ecase{y}{C}{\dc_i}{\overline{x}}{e_i}
  \\[0.03in]
\end{array}
$$

\emphbf{Reductions}\hfill{{$\fbox{\goesto{\prog}{\prog'}}$}}
$$
\begin{array}{rcl}
C[\prog]
  & \hookrightarrow
  & C[\prog'],\quad \text{if}\ \goesto{\prog}{\prog'}
  \\
{c\ v}
  & \hookrightarrow
  & {\ceval{c}{v}}
  \\
{({\efun{x}{\typ}{e})}\ {e'}}
  & \hookrightarrow
  & {\SUBST{e}{x}{e'}}
  \\
{\ecase{y}{\dc_j\ \overline{e}}{\dc_i}{\overline{x_i}}{e_i}}
  & \hookrightarrow
  & {\SUBST{\SUBST{e_j}{y}{\dc_j\ \overline{e}}}{\overline{x_i}}{\overline{e}}}
  \\
{\erefb{x}{\gtyp}{e}{\prog}}
  & \hookrightarrow
  & {\SUBST{\prog}{x}{\efix{(\efun{x}{\gtyp}{e})}}}
  \\
{\eletb{x}{\gtyp}{\bd_x}{\bd}}
  & \hookrightarrow
  % & {\SUBST{\prog}{x}{e}}
  & {\SUBST{\bd}{x}{\efix{(\efun{x}{\gtyp}{\bd_x})}}}
  \\
%% {\eletrec{x}{e_x}{\gtyp}{[L]}{\prog}}
  %% & \hookrightarrow
  %% & {\SUBST{\prog}{x}{\efix{\efun{x}{\gtyp}{e_x}}}}
  %% \\
{\efix{\prog}}
  & \hookrightarrow
  & {\prog\ (\efix{\prog})}
\end{array}
$$
\caption{\textbf{Operational Semantics of \corelan}}
\label{fig:semantics}
\end{figure}

Figure~\ref{fig:syntax} summarizes the syntax of \corelan,
which is essentially the calculus \undeclang~\cite{Vazou14}
with explicit recursion and a special $\erefname$ binding form
to denote terms that are reflected into the refinement logic.
In \corelan refinements $r$ are arbitrary expressions $e$
(hence $r ::= e$ in Figure~\ref{fig:syntax}).
This choice allows us to prove preservation and progress,
but renders typechecking undecidable.
In \S~\ref{sec:algorithmic} we will see how to recover
decidability by soundly approximating refinements.

The syntactic elements of \corelan are layered into
primitive constants, values, expressions, binders
and programs.

\mypara{Constants}
The primitive constants of \corelan
include all the primitive logical
operators $\op$, here, the set $\{ =, <\}$.
Moreover, they include the
primitive booleans $\etrue$, $\efalse$,
integers $\mathtt{-1}, \mathtt{0}$, $\mathtt{1}$, \etc,
and logical operators $\mathtt{\land}$, $\mathtt{\lor}$, $\mathtt{\lnot}$, \etc.

\mypara{Data Constructors}
We encode data constructors as special constants.
% Each data type has an equality predicate $\haseq{T}$
% that is true only if values of type $T$ can be finitely compared.
For example the data type \tintlist, which represents
finite lists of integers, has two data constructors: $\dnull$ (``nil'')
and $\dcons$ (``cons'').
% and satisfies $\haseq{\tintlist}$.

%% NV Arity is not used anywhere
%%Each data type has an arity $\arity{T}$ that represents
%%the exact number of data constructors that return
%%a value of type $T$.
%
%%For example the data type \tintlist, which represents
%%lists of integers, has two data constructors: $\dnull$ (``nil'')
%%and $\dcons$ (``cons'') and so has arity $2$.

\mypara{Values \& Expressions}
The values of \corelan include
constants, $\lambda$-abstractions
$\efun{x}{\typ}{e}$, and fully
applied data constructors $D$
that wrap values.
The expressions of \corelan
include values and variables $x$,
applications $\eapp{e}{e}$, and
$\mathtt{case}$ expressions.

\mypara{Binders \& Programs}
A \emph{binder} $\bd$ is a series of possibly recursive
let definitions, followed by an expression.
A \emph{program} \prog is a series of $\erefname$
definitions, each of which names a function
that can be reflected into the refinement
logic, followed by a binder.
The stratification of programs via binders
is required so that arbitrary recursive definitions
are allowed but cannot be inserted into the logic
via refinements or reflection.
(We \emph{can} allow non-recursive $\mathtt{let}$
binders in $e$, but omit them for simplicity.)

\subsection{Operational Semantics}

Figure~\ref{fig:syntax} summarizes the small
step contextual $\beta$-reduction semantics for
\corelan.
%
%%There are two points to note.
%%%
%%First, we allow for reductions under
%%data constructors, and thus, values may
%%be further reduced.
%%%
%%Second, for simplicity, we treat both
%%$\eletname$ and $\erefname$ as possibly
%%recursive (\ie $\mathtt{let\ rec}$) binders.
%% Note that, for simplicity, we treat each
%% $\eletname$ as a possibly
%% recursive (\ie $\mathtt{let\ rec}$) binder.
%
We write \evalj{e}{e'}{j} if there exist
$e_1,\ldots,e_j$ such that $e$ is $e_1$,
$e'$ is $e_j$ and $\forall i,j, 1 \leq i < j$,
we have $\evals{e_i}{e_{i+1}}$.
We write $\evalsto{e}{e'}$ if there exists
some finite $j$ such that $\evalj{e}{e'}{j}$.
We define $\betaeq{}{}$ to be the reflexive,
symmetric, transitive closure of $\evals{}{}$.

\mypara{Constants} Application of a constant requires the
argument be reduced to a value; in a single step the
expression is reduced to the output of the primitive
constant operation.
For example, consider $=$, the primitive equality
operator on integers.
We have $\ceval{=}{n} \defeq =_n$
where $\ceval{=_n}{m}$ equals \etrue
iff $m$ is the same as $n$.
%
%\mypara{Equality}
%
We assume that the equality operator
is defined \emph{for all} values,
and, for functions, is defined as
extensional equality.
That is, for all
$f$ and
$f'$
we have
$\evals{(f = f')}{\etrue}
  \quad \mbox{iff} \quad
  \forall v.\ \betaeq{f\ v}{f'\ v}$.
We assume source \emph{terms} only contain implementable equalities
over non-function types; the above only appears in \emph{refinements}
and allows us to state and prove facts about extensional
equality~\S~\ref{subsec:extensionality}.
%% % \RJ{CHECK}

%%That is, for all
%%$f \defeq \efun{x}{\typ}{e}$ and
%%$f' \defeq \efun{x}{\typ}{e'}$
%%we have
%%$$\evals{(f = f')}{\etrue}
%%  \quad \mbox{iff} \quad
%%  \forall v.\ \evalsto{(\SUBST{e}{x}{v} = \SUBST{e'}{x}{v})}{\etrue}
%%$$

\subsection{Types}

\corelan types include basic types, which are \emph{refined} with predicates,
and dependent function types.
\emph{Basic types} \btyp comprise integers, booleans, and a family of data-types
$T$ (representing lists, trees \etc.)
For example the data type \tintlist represents lists of integers.
We refine basic types with predicates (boolean valued expressions \refa) to obtain
\emph{basic refinement types} $\tref{v}{\btyp}{\refa}$.
Finally, we have dependent \emph{function types} $\tfun{x}{\typ_x}{\typ}$
where the input $x$ has the type $\typ_x$ and the output $\typ$ may
refer to the input binder $x$.
We write $\btyp$ to abbreviate $\tref{v}{\btyp}{\etrue}$,
and \tfunbasic{\typ_x}{\typ} to abbreviate \tfun{x}{\typ_x}{\typ} if
$x$ does not appear in $\typ$.
We use $r$ to refer to refinements.

\mypara{Denotations}
Each type $\typ$ \emph{denotes} a set of expressions $\interp{\typ}$,
that are defined via the dynamic semantics~\cite{Knowles10}.
Let $\shape{\typ}$ be the type we get if we erase all refinements
from $\typ$ and $\bhastype{}{e}{\shape{\typ}}$ be the
standard typing relation for the typed lambda calculus.
Then, we define the denotation of types as:
\begin{align*}
\interp{\tref{x}{\btyp}{r}} \defeq &
    \{e \mid  \bhastype{}{e}{\btyp},
              \mbox{ if } \evalsto{e}{w}
              \mbox{ then } \evalsto{r\subst{x}{w}}{\etrue} \}\\
\interp{\tfun{x}{\typ_x}{\typ}} \defeq &
    \{e \mid  \bhastype{}{e}{\shape{\tfunbasic{\typ_x}{\typ}}},
              \forall e_x \in \interp{\typ_x}.\ \eapp{e}{e_x} \in \interp{\typ\subst{x}{e_x}}
    \}
\end{align*}

\mypara{Constants}
For each constant $c$ we define its type \constty{c}
such that $c \in \interp{\constty{c}}$.
For example,
$$
\begin{array}{lcl}
\constty{3} &\doteq& \tref{v}{\tint}{v = 3}\\
\constty{+} &\doteq& \tfun{\ttx}{\tint}{\tfun{\tty}{\tint}{\tref{v}{\tint}{v = x + y}}}\\
\constty{\leq} &\doteq& \tfun{\ttx}{\tint}{\tfun{\tty}{\tint}{\tref{v}{\tbool}{v \Leftrightarrow x \leq y}}}\\
\end{array}
$$
So, by definition we get the constant typing lemma
\begin{lemma}{[Constant Typing]}\label{lemma:constants}
Every constant $c \in \interp{\constty{c}}$.
\end{lemma}
Thus, if $\constty{c} \defeq \tfun{x}{\typ_x}{\typ}$,
then for every value $w \in \interp{\typ_x}$, we require
$\ceval{c}{w} \in \interp{\typ\subst{x}{w}}$.

%% \mypara{Equality}
%% The equality predicate
%% $\haseq{B}$, is defined to be true for \tint and \tbool,
%% and for each type constructor $T$
%% whose values can be finitely compared.
%% %
%% So, by definition we get the equality lemma
%% %
%% \begin{lemma}{[Equality]}\label{lemma:equality}
%% If $\haseq{B}$ then for each value $\bhastype{\emptyset}{w}{B}$
%% \evalsto{w = w}{\etrue}
%% \end{lemma}

\subsection{Refinement Reflection}
\label{subsec:logicalannotations}
The simple, but key idea in our work is to
\emph{strengthen} the output type of functions
with a refinement that \emph{reflects} the
definition of the function in the logic.
We do this by treating each
$\erefname$-binder:
${\erefb{f}{\gtyp}{e}{\prog}}$
as a $\eletname$-binder:
${\eletb{f}{\exacttype{\gtyp}{e}}{e}{\prog}}$
during type checking (rule $\rtreflect$ in Figure~\ref{fig:typing}).

\mypara{Reflection}
We write \exacttype{\typ}{e} for the \emph{reflection}
of term $e$ into the type $\typ$,  defined by strengthening
\typ as:
$$
\begin{array}{lcl}
\exacttype{\tref{v}{\btyp}{r}}{e}
  & \defeq
  & \tref{v}{\btyp}{r \land v = e}\\
\exacttype{\tfun{x}{\typ_x}{\typ}}{\efun{y}{}{e}}
  & \defeq
  & \tfun{x}{\typ_x}{\exacttype{\typ}{e\subst{y}{x}}}
\end{array}
$$
As an example, recall from \S~\ref{sec:overview}
that the @reflect fib @ strengthens the type of
@fib@ with the reflected refinement @fibP@.
%% NV In Overview, we have fibP v n = v = fib n && fibR v n
%% NV Here we get the reflection part (fibR v n)
%% NV which we can verify
%% NV at each fix invocation we also get the v = fib n portion
%% NV via the exact rule
%% NV We can not add the v = fib n as a port condition, because
%% NV we cannot prove it.

\mypara{Consequences for Verification}
Reflection has two consequences for verification.
First, the reflected refinement is \emph{not trusted};
it is itself verified (as a valid output type)
during type checking.
Second, instead of being tethered to quantifier
instantiation heuristics or having to program
``triggers'' as in Dafny~\citep{dafny} or
\fstar~\citep{fstar}
the programmer can predictably ``unfold'' the
definition of the function during a proof simply
by ``calling'' the function, which we have found
to be a very natural way of structuring
proofs~\S~\ref{sec:evaluation}.

\subsection{Refining \& Reflecting Data Constructors with Measures}
\label{subsec:measures}
\label{subsec:list}

% We reuse the notion of \emph{measures}~\cite{Vazou14}
% to reflect functions over datatypes into the refinement
% logic.

We assume that each data type is equipped with
a set of \emph{measures} which are \emph{unary}
functions whose (1)~domain is the data type, and
(2)~body is a single case-expression over the
datatype~\cite{Vazou14}:
$$\emeasb{f}
         {\gtyp}
         {\efun{x}{\typ}{\ecase{y}{x}{\dc_i}{\overline{z}}{e_i}}}$$
For example, @len@ measures the size of an $\tintlist$:
\begin{code}
  measure len :: [Int] -> Nat
  len = \x -> case x of
                []     -> 0
                (x:xs) -> 1 + len xs
\end{code}

\mypara{Checking and Projection}
We assume the existence of measures that
\emph{check} the top-level constructor,
and \emph{project} their individual fields.
\NV{Remove this pointer since we removed the text}
In \S~\ref{subsec:embedding} we show how to
use these measures to reflect functions over
datatypes.
%
% Such measures are straightforward to generate
% automatically from the data-type definition.)
%
For example, for lists, we assume the existence of measures:
\begin{code}
  isNil []      = True
  isNil (x:xs)  = False

  isCons (x:xs) = True
  isCons []     = False

  sel1 (x:xs)   = x
  sel2 (x:xs)   = xs
\end{code}

\mypara{Refining Data Constructors with Measures}
We use measures to strengthen the types
of data constructors, and we use these
strengthened types during construction
and destruction (pattern-matching).
Let:
(1)~$\dc$ be a data constructor,
   with \emph{unrefined} type
   $\tfun{\overline{x}}{\overline{\gtyp}}{T}$
(2)~the $i$-th measure definition with
   domain $T$ is:
$$\emeasb{f_i}
         {\gtyp}
         {\efun{x}{\typ}{\ecase{y}{x}{\dc}{\overline{z}}{e_{i}}}}
$$
Then, the refined type of $\dc$ is defined:
$$
\constty{\dc} \defeq
   \tfun{\overline{x}}
        {\overline{\typ}}
        {\tref{v}{T}{ \wedge_i f_i\ v = \SUBST{e_{i}}{\overline{z}}{\overline{x}}}}
$$

Thus, each data constructor's output type is refined to reflect
the definition of each of its measures.
For example, we use the measures @len@, @isNil@, @isCons@, @sel1@,
and @sel2@ to strengthen the types of $\dnull$ and $\dcons$ to:
\begin{align*}
\constty{\dnull}  \defeq & \tref{v}{\tintlist}{r_{\dnull}} \\
\constty{\dcons}  \defeq & \tfun{x}{\tint}
                                   {\tfun{\mathit{xs}}
                                         {\tintlist}
                                         {\tref{v}{\tintlist}{r_\dcons}}}
\intertext{where the output refinements are}
r_{\dnull} \defeq &\ \mathtt{len}\ v = 0
             \land  \mathtt{isNil}\ v
             \land  \lnot \mathtt{isCons}\ v \\
r_{\dcons} \defeq &\ \mathtt{len}\ v = 1 + \mathtt{len}\ \mathit{xs}
             \land  \lnot \mathtt{isNil}\ v
             \land  \mathtt{isCons}\ v \\
             \land & \  \mathtt{sel1}\ v = x
             \land  \mathtt{sel2}\ v = \mathit{xs}
\end{align*}
It is easy to prove that Lemma~\ref{lemma:constants}
holds for data constructors, by construction.
For example, $\mathtt{len}\ \dnull = 0$ evaluates to $\tttrue$.

%%
%% The above annotation \emph{strengthens} the types of data constructors
%% $\dc_i$ to reflect the
%% behavior of $f$:
%%
%%
%% \preproc{\eletrecoptsmall{f}{\efun{x}{\typ}{\ecase{y}{x}{\dc_i}{\overline{z}}{e_i}}}{\gtyp}{M}{\prog}}
%%
%% %
%% Where \exacttypefun{f}{\gtyp}{e} strengthens the result $v$ of the type \typ to exactly
%% describe that $f\ v = e$:
%% %
%% \begin{align*}
%% \exacttypefun{f}{\tref{v}{\btyp}{r}}{e} &= \tref{v}{\btyp}{r \land f\ v = e}\\
%% \exacttypefun{f}{\tfun{x}{\typ_x}{\typ}}{\efun{y}{}{e}} &=\tfun{x}{\typ_x}{\exacttype{\typ}{e\subst{y}{x}}}
%% \end{align*}

\subsection{Typing Rules}
\begin{figure}
\emphbf{Well Formedness}\hfill{\fbox{\iswellformed{\env}{\typ}}}\\

$$
\inference{
	\hastype{\env,\tbind{v}{\btyp}}{\refa}{\tbool}
}{
	\iswellformed{\env}{\tref{v}{\btyp}{\refa}}
}[\rwbase]
\inference{
	\iswellformed{\env}{\typ_x} &
	\iswellformed{\env,\tbind{x}{\typ_x}}{\typ}
}{
	\iswellformed{\env}{\tfun{x}{\typ_x}{\typ}}
}[\rwfun]
$$

\emphbf{Subtyping}\hfill{\fbox{\issubtype{\env}{\typ_1}{\typ_2}}}\\

$$
\inference{
	\forall \sub\in\interp{\env}.
	\interp{\applysub{\sub}{\tref{v}{\btyp}{\refa_1}}}
	\subseteq
	\interp{\applysub{\sub}{\tref{v}{\btyp}{\refa_2}}}
}{
	\issubtype{\env}{\tref{v}{\btyp}{\refa_1}}{\tref{v}{\btyp}{\refa_2}}
}[\rsubbase]
$$

$$
\inference{
	\issubtype{\env}{\typ_x'}{\typ_x} &
	\issubtype{\env,\tbind{x}{\typ_x'}}{\typ}{\typ'}
}{
	\issubtype{\env}{\tfun{x}{\typ_x}{\typ}}{\tfun{x}{\typ_x'}{\typ'}}
}[\rsubfun]
$$

\emphbf{Typing}\hfill{\fbox{\hastype{\env}{\prog}{\typ}}}\\
$$
\inference{
	\tbind{x}{\gtyp}\in\env
}{
	\hastype{\env}{x}{\gtyp}
}[\rtvar]
\qquad
\inference{
}{
	\hastype{\env}{c}{\constty{c}}
}[\rtconst]
$$

$$
\inference{
	\hastype{\env}{\prog}{\typ'} &
	\issubtype{\env}{\typ'}{\typ}
}{
	\hastype{\env}{\prog}{\typ}
}[\rtsub]
$$
$$
\inference{
    % \haseq{\btyp} &
	\hastype{\env}{e}{\tref{v}{\btyp}{\reft_r}}
}{
	\hastype{\env}{e}{\tref{v}{\btyp}{\reft_r\land v = e}}
}[\rtexact]
$$
$$
\inference{
	\hastype{\env, \tbind{x}{\typ_x}}{e}{\typ}
}{
	\hastype{\env}{\efun{x}{\typ}{e}}{\tfun{x}{\typ_x}{\typ}}
}[\rtfun]
$$

$$
\inference{
	\hastype{\env}{e_1}{(\tfun{x}{\typ_x}{\typ})} &&
	\hastype{\env}{e_2}{\typ_x}
}{
	\hastype{\env}{e_1\ e_2}{\typ}
}[\rtapp]
$$

%%% $$
%%% \inference{
	%%% \hastype{\env}{e}{\gtyp_x} &
	%%% \hastype{\env, \tbind{x}{\gtyp_x}}{\prog}{\typ} &
	%%% \iswellformed{\env}{\typ}
%%% }{
	%%% \hastype{\env}{\elet{x}{e}{\gtyp_x}{}{\prog}}{\typ}
%%% }[\rtlet]
%%% $$
%%%
%%% $$
%%% \inference{
	%%% \hastype{\env, \tbind{x}{\gtyp_x}}{e}{\gtyp_x} &
	%%% \hastype{\env, \tbind{x}{\gtyp_x}}{\prog}{\typ} &
	%%% \iswellformed{\env}{\typ}
%%% }{
	%%% \hastype{\env}{\eletrec{x}{e}{\gtyp_x}{}{\prog}}{\typ}
%%% }[\rtletrec]
%%% $$
%% \NV{For soundness, it is important that f cannot appear in $\gtyp_2$}
$$
\inference
	{\hastype{\env, \tbind{x}{\gtyp_x}}{\bd_x}{\gtyp_x} &
	 \iswellformed{\env, \tbind{x}{\gtyp_x}}{\typ_x} \\
	 \hastype{\env, \tbind{x}{\gtyp_x}}{\bd}{\gtyp} &
	 \iswellformed{\env}{\typ}}
	{\hastype{\env}{\eletb{x}{\gtyp_x}{\bd_x}{\bd}}{\typ}}
	[\rtlet]
$$

$$
\inference
	{\hastype{\env}
	 				 {\eletb{f}{\exacttype{\gtyp_f}{e}}{e}{\prog}}
					 {\typ}
	}
	{\hastype{\env}
					 {\erefb{f}{\gtyp_f}{e}{\prog}}
					 {\typ}
	}[\rtreflect]
$$

$$
\inference{
	\hastype{\env}{e}{\tref{v}{T}{e_r}} & \iswellformed{\env}{\typ} \\
	& \forall i. \constty{\dc_i} = \tfunbasic{\overline{\tbind{y_j}{\typ_j}}}{\tref{v}{T}{e_{r_i}}} \\
	& \hastype{\env, \overline{\tbind{y_j}{\typ_j}}, \tbind{x}{\tref{v}{T}{e_r \land e_{r_i}}} }{e_i}{\typ}
}{
	\hastype{\env}{\ecase{x}{e}{\dc_i}{\overline{y_i}}{e_i}}{\typ}
}[\rtcase]
$$
$$
\inference{
	\iswellformed{\env}{\typ}
}{
	\hastype{\env}
	        {\efix{}}
	        {\tfun{x}{\typ}{\tfun{y}{\typ}{\typ}}}
}[\rtfix]
$$
\caption{Typing of \corelan}
\label{fig:typing}
\end{figure}

Next, we present the type-checking
judgments and rules of \corelan.

\mypara{Environments and Closing Substitutions}
A \emph{type environment} $\env$ is a sequence of type bindings
$\tbind{x_1}{\typ_1},\ldots,\tbind{x_n}{\typ_n}$. An environment
denotes a set of \emph{closing substitutions} $\sto$ which are
sequences of expression bindings:
$\gbind{x_1}{e_1}, \ldots, \gbind{x_n}{e_n}$ such that:
$$
\interp{\env} \defeq  \{\sto \mid \forall \tbind{x}{\typ} \in \Env.
                                    \sto(x) \in \interp{\applysub{\sto}{\typ}} \}
$$

\mypara{Judgments}
We use environments to define three kinds of
rules: Well-formedness, Subtyping,
and Typing~\citep{Knowles10,Vazou14}.
%
%\mypara{Well-formedness}
A judgment \iswellformed{\env}{\typ} states that
the refinement type $\typ$ is well-formed in
the environment $\env$.
Intuitively, the type $\typ$ is well-formed if all
the refinements in $\typ$ are $\tbool$-typed in $\env$.
%
%\mypara{Subtyping}
A judgment \issubtype{\env}{\typ_1}{\typ_2} states
that the type $\typ_1$ is a subtype of %the type
$\typ_2$ in the environment $\env$.
Informally, $\typ_1$ is a subtype of $\typ_2$ if, when
the free variables of $\typ_1$ and $\typ_2$
are bound to expressions described by $\env$,
the denotation of $\typ_1$
is \emph{contained in} the denotation of $\typ_2$.
Subtyping of basic types reduces to denotational containment checking.
%
%\mypara{Implication}
%%A judgment \issubref{\Env}{p_1}{p_2} states
%%that the predicate $p_1$ \emph{implies}
%%the predicate $p_2$ in the environment $\Env$.
%
That is, for any closing substitution $\sto$
in the denotation of $\env$, for every expression $e$,
if $e \in \interp{\applysub{\sto}{\typ_1}}$ then
$ e \in \interp{\applysub{\sto}{\typ_2}}$.
%
%\mypara{Typing}
A judgment \hastype{\env}{\prog}{\typ} states that
the program $\prog$ has the type $\typ$ in
the environment $\env$.
That is, when the free variables in $\prog$ are
bound to expressions described by $\env$, the
program $\prog$ will evaluate to a value
described by $\typ$.

\mypara{Rules}
All but three of the rules are standard~\cite{Knowles10,Vazou14}.
First, rule \rtreflect is used to strengthen the type of each
reflected binder with its definition, as described previously
in \S~\ref{subsec:logicalannotations}.
%
% \NV{FIX:Eq}
% applies only to expressions that can be finitely compared
% (\ie whose type satisfies the \haseq{B} predicate) and
Second, rule \rtexact strengthens the expression with
a singleton type equating the value and the expression
(\ie reflecting the expression in the type).
This is a generalization of the ``selfification'' rules
from \cite{Ou2004,Knowles10}, and is required to
equate the reflected functions with their definitions.
For example, the application $(\fib\ 1)$ is typed as
${\tref{v}{\tint}{\fibdef\ v\ 1 \wedge v = \fib\ 1}}$ where
the first conjunct comes from the (reflection-strengthened)
output refinement of \fib~\S~\ref{sec:overview}, and
the second conjunct comes from rule~\rtexact.
Finally, rule \rtfix is used to type the intermediate
$\texttt{fix}$ expressions that appear, not in the
surface language but as intermediate terms in the
operational semantics.

\mypara{Soundness}
Following \undeclang~\citep{Vazou14}, we can show that
evaluation preserves typing and that typing implies
denotational inclusion.
\begin{theorem}{[Soundness of \corelan]}\label{thm:safety}
\begin{itemize}
\item\textbf{Denotations}
If $\hastype{\env}{\prog}{\typ}$ then
$\forall \sto\in \interp{\env}. \applysub{\sto}{\prog} \in \interp{\applysub{\sto}{\typ}}$.
\item\textbf{Preservation}
If \hastype{\emptyset}{\prog}{\typ}
       and $\evalsto{\prog}{w}$ then $\hastype{\emptyset}{w}{\typ}$.
\end{itemize}
\end{theorem}

\subsection{From Programs \& Types to Propositions \& Proofs}

The denotational soundness Theorem~\ref{thm:safety}
lets us interpret well typed programs as proofs of
propositions.

\NV{say that definition is a context that will be used later}
\mypara{``Definitions''}
A \emph{definition} $\defn$ is a sequence of reflected binders:
$$\defn \ ::= \ \bullet \spmid \erefb{x}{\gtyp}{e}{\defn}$$
A \emph{definition's environment} $\env(\defn)$ comprises
its binders and their reflected types:
\begin{align*}
\aenv(\bullet)                    \defeq & \emptyset \\
\aenv(\erefb{f}{\gtyp}{e}{\defn}) \defeq & (f, \exacttype{\gtyp}{e}),\ \env(\defn) \\
\end{align*}
A \emph{definition's substitution} $\sto(\defn)$ maps each binder
to its definition:
\begin{align*}
\sto(\bullet)                     \defeq & \emptysto \\
\sto(\erefb{f}{\gtyp}{e}{\defn})  \defeq & \extendsto{f}{\efix{f}\ e}{\sto(\defn)}
\end{align*}

\mypara{``Propositions''}
A \emph{proposition} is a type
$$\tbind{x_1}{\typ_1} \rightarrow \ldots
  \rightarrow \tbind{x_n}{\typ_n}
  \rightarrow \tref{v}{\tunit}{\ppn}$$
For brevity, we abbreviate propositions like the above to
$\tfun{\overline{x}}{\overline{\typ}}{\ttref{\ppn}}$
and we call $\ppn$ the \emph{proposition's refinement}.
For simplicity we assume that $\freevars{\typ_i} = \emptyset$.

\mypara{``Validity''}
\NV{add termination: proofs should provably terminates}

A proposition $\tfun{\overline{x}}{\overline{\typ}}{\ttref{\ppn}}$
is \emph{valid under} $\defn$ if
$$\forall \overline{w} \in \interp{\overline{\typ}}.\
  \evalsto{\applysub{\sto(\defn)}{\SUBST{\ppn}{\overline{x}}{\overline{w}}}}{\etrue}$$
That is, the proposition is valid if its refinement
evaluates to $\etrue$ for every (well typed)
interpretation for its parameters $\overline{x}$
under $\defn$.

\mypara{``Proofs''}
A binder $\bd$ \emph{proves} a proposition $\gtyp$ under $\defn$ if
$$\hastype{\emptyset}{\defn[\eletb{x}{\typ}{\bd}{\eunit}]}{\tunit}$$
That is, if the binder $\bd$ has the proposition's type $\gtyp$
under the definition $\defn$'s environment.

\begin{theorem}{[Proofs]} \label{thm:validity}
If $\bd$ proves $\typ$ under $\defn$ then $\typ$ is valid under $d$.
\end{theorem}

\begin{proof}
As $\bd$ proves $\typ$ under $\defn$, we have
\begin{align}
\hastype{\emptyset}{\defn[\eletb{x}{\typ}{\bd}{\eunit}]}{\tunit}
\label{pf:1} \\
\intertext{By Theorem~\ref{thm:safety} on \ref{pf:1} we get}
\sto(\defn) \in \interp{\env(\defn)} \label{pf:2}\\
\intertext{Furthermore,  by the typing rules \ref{pf:1}
implies $\hastype{\env(\defn)}{\bd}{\typ}$ and hence, via Theorem~\ref{thm:safety}}
\forall \sub \in \interp{\env(\defn)}.\ \applysub{\sub}{\bd} \in \interp{\applysub{\sub}{\typ}}
\label{pf:3} \\
\intertext{Together, \ref{pf:2} and \ref{pf:3} imply}
\applysub{\sto(\defn)}{\bd} \in \interp{\applysub{\sto(\defn)}{\typ}}
\label{pf:4}
\intertext{By the definition of type denotations, we have}
\interp{\applysub{\sto(\defn)}{\typ}}
  \defeq \{ f\ |\ \typ \mbox{ is valid under}\ \defn\}
  \label{pf:5}
\end{align}
By \ref{pf:4}, the above set is not empty, and hence $\typ$ is valid under $\defn$.
\end{proof}

\mypara{Example: Fibonacci is increasing}
In \S~\ref{sec:overview} we verified that
under a definition $\defn$ that includes \fibname,
the term \fibincrname proves
$${\tfun{n}{\tnat}{\ttref{\fib{n} \leq \fib{(n+1)}}}}$$
Thus, by Theorem~\ref{thm:validity} we get
%that the \fib{}, as defined in $\defn$ is increasing
%
%$$
%\forall n\in\interp{\tnat}. \evalsto{\fib{n} \leq \fib{(n+1)}}{\etrue}
%$$
%Equivalently
$$
\forall n. \evalsto{0 \leq n}{\etrue} \Rightarrow \evalsto{\fib{n} \leq \fib{(n+1)}}{\etrue}
$$

\section{Algorithmic Verification}\label{sec:algorithmic}

Next, we describe \smtlan, a conservative approximation
of \corelan where the undecidable type subsumption rule
is replaced with a decidable one, yielding an SMT-based
algorithmic type system that enjoys the same soundness
guarantees.

\subsection{The SMT logic \smtlan}

\begin{figure}[t!]
\centering
$$
\begin{array}{rrcl}
\emphbf{Predicates} \quad
  & \pred & ::= &
    \pred \binop \pred \spmid
    \unop \pred \\
  && \spmid & n \spmid b \spmid x \spmid \dc \spmid  x\ \overline{\pred}\\
%%  && \spmid & \forall \overline{\tbind{x}{\sort}}. \pred
  && \spmid & \eif{\pred}{\pred}{\pred}
\\[0.03in]

\emphbf{Integers} \quad
  & n
  & ::= & 0, -1, 1, \dots
\\[0.03in]

\emphbf{Booleans} \quad
  & b
  & ::= & \etrue \spmid \efalse
\\[0.03in]

\emphbf{Bin Operators} \quad
  & \binop
  & ::= & = \spmid < \spmid \land \spmid + \spmid - \spmid \dots
\\[0.03in]

\emphbf{Un Operators} \quad
  & \unop
  & ::= & \lnot \spmid \dots 
\\[0.03in]

\emphbf{Model} \quad
  & \sigma
  & ::= & \sigma, (x:\pred) \spmid \emptyset
\\[0.03in]

\emphbf{Sort Arguments} \quad
  & \sort_a
  & ::= & \tint \spmid \tbool \spmid \tuniv 
         \spmid \tsmtfun{\sort_a}{\sort_a}
\\[0.03in]
\emphbf{Sorts} \quad
  & \sort
  & ::=  & \sort_a \rightarrow \sort
\end{array}
$$
\caption{\textbf{Syntax of \smtlan}}
\label{fig:smtsyntax}
\end{figure}

\mypara{Syntax: Terms \& Sorts}
Figure~\ref{fig:smtsyntax} summarizes the syntax
of \smtlan, the \emph{sorted} (SMT-)
decidable logic of quantifier-free equality,
uninterpreted functions and linear
arithmetic (QF-EUFLIA) ~\citep{Nelson81,SMTLIB2}.
The \emph{terms} of \smtlan include
integers $n$,
booleans $b$,
variables $x$,
data constructors $\dc$ (encoded as constants),
fully applied unary \unop and binary \binop operators,
and application $x\ \overline{\pred}$ of an uninterpreted function $x$.
The \emph{sorts} of \smtlan include built-in
integer \tint and \tbool for representing
integers and booleans.
%
%% NV reflected functions and measures are first order
%% NV because
%% NV 1. they can be partially applied
%% NV 2. they can be passed as arguments
The interpreted functions of \smtlan, \eg
the logical constants $=$ and $<$,
%% NV and the uninterpreted functions app and lam
%% NV but we have not introduced these yet
have the function sort $\sort \rightarrow \sort$.
Other functional values in \corelan, \eg
reflected \corelan functions and
$\lambda$-expressions, are represented as
first-order values with
uninterpreted sort \tsmtfun{\sort}{\sort}.
%
%%The uninterpreted functions of \smtlan, which
%%correspond to reflected \corelan functions,
%%have the function sort $\sort \rightarrow \sort$.
%%%
%%Other functional values in \corelan, \eg
%%$\lambda$-expressions, are represented as
%%first-order values in \smtlan with
%%uninterpreted sort \tsmtfun{\sort}{\sort}.
%%%
The universal sort \tuniv represents all other values.

\mypara{Semantics: Satisfaction \& Validity}
An assignment $\sigma$ is a mapping from
variables to terms
${\sigma \defeq \{ \assignto{x_1}{\pred_1}, \ldots, \assignto{x_n}{\pred_n} \}}$.
We write
${\sigma \models \pred}$
if the assignment $\sigma$ is a
\emph{model of} $\pred$, intuitively
if $\sigma\ \pred$ ``is true''~\cite{Nelson81}.
A predicate $\pred$ \emph{is satisfiable} if
there exists ${\sigma\models\pred}$.
A predicate $\pred$ \emph{is valid} if
for all assignments ${\sigma\models\pred}$.

\subsection{Transforming \corelan into \smtlan}
\label{subsec:embedding}

\newcommand\emptyaxioms{\ensuremath{\emptyset}\xspace}
\newcommand\andaxioms[2]{\ensuremath{{#1}\cup {#2}}\xspace}

\begin{figure}
\emphbf{Transformation}\hfill{\fbox{\tologicshort{\Gamma}{e}{\typ}{\pred}{\sort}{\smtenv}{\axioms}}}
$$
\inference{
}{
	\tologicshort{\env}{b}{\tbool}{b}{\tbool}{\emptyset}{\emptyaxioms}
}[\lgbool]
\qquad
\inference{
}{
	\tologicshort{\env}{n}{\tint}{n}{\tint}{\emptyset}{\emptyaxioms}
}[\lgint]
$$

$$
\inference{
    \tologicshort{\env}{e_1}{\typ}{\pred_1}{\embed{\typ}}{\smtenv}{\axioms_1} &
    \tologicshort{\env}{e_2}{\typ}{\pred_2}{\embed{\typ}}{\smtenv}{\axioms_2}
}{
	\tologicshort{\env}{e_1\binop e_2}{\tbool}{\pred_1 \binop\pred_2}{\tbool}{\smtenv}{\andaxioms{\axioms_1}{\axioms_2}}
}[\lgbinGEN]
$$

$$
\inference{
	\tologicshort{\env}{e}{\tbool}{\pred}{\tbool}{\smtenv}{\axioms}
}{
	\tologicshort{\env}{\unop e}{\tbool}{\unop\pred}{\tbool}{\smtenv}{\axioms}
}[\lgun]
\qquad
\inference{
}{
	\tologicshort{\env}{x}{\env(x)}{x}{\embed{\env(x)}}{\emptyset}{\emptyaxioms}
}[\lgvar]
$$

$$
\inference{
}{
	\tologicshort{\env}{c}{\constty{\odot}}{\smtvar{c}}{\embed{\constty{\odot}}}{\emptyset}{\emptyaxioms}
}[\lgpop]
\qquad
\inference{
}{
	\tologicshort{\env}{\dc}{\constty{\dc}}{\smtvar{\dc}}{\embed{\constty{\dc}}}{\emptyset}{\emptyaxioms}
}[\lgdc]
$$

%%$$
%%\inference{
%%  	\axioms_{f_1} = \forall \tbind{x}{\sort_x}.\smtappname{\sort_x}{\sort}\ f\ x = \pred \\
%%  	\axioms_{f_2} = \forall \tbind{g}{\sort'},\tbind{x}{\sort_x}.
%%  	\smtappname{\sort_x}{\sort}\ f\ x = \smtappname{\sort_x}{\sort}\ g\ x \Rightarrow f = g \\
%% 	f\ \text{fresh} &
%% 	\sort' = \embed{\tfun{x}{\typ_x}{\typ}} &
%% 	\sort  = \embed{\typ} &
%% 	\sort_x = \embed{\typ_x} \\
%% 	\tologicshort{\env,\tbind{x}{\typ_x}}{e}{\typ}{\pred}{\sort}{\smtenv, \tbind{x}{\sort_x}}{\axioms} &
%% 	\hastype{\env}{(\efun{x}{}{e})}{(\tfun{x}{\typ_x}{\typ})}\\
%%}{
%%	\tologicshort{\env}{\efun{x}{}{e}}{(\tfun{x}{\typ_x}{\typ})}
%%	        {f}{\sort'}{\smtenv, \tbind{f}{\sort'}}{\andaxioms{\{\axioms_{f_1}, \axioms_{f_2}\}}{\axioms}}
%%}[\lgfun]
%%$$

$$
\inference{
    \tologicshort{\env, \tbind{x}{\typ_x}}{e}{}{\pred}{}{}{} &
  	\hastype{\env}{(\efun{x}{}{e})}{(\tfun{x}{\typ_x}{\typ})}\\
}{
	\tologicshort{\env}{\efun{x}{}{e}}{(\tfun{x}{\typ_x}{\typ})}
	        {\smtlamname{\embed{\typ_x}}{\embed{\typ}}\ {x}\ {\pred}}
	        {\sort'}{\smtenv, \tbind{f}{\sort'}}{\andaxioms{\{\axioms_{f_1}, \axioms_{f_2}\}}{\axioms}}
}[\lgfun]
$$

$$
\inference{
	\tologicshort{\env}{e'}{\typ_x}{\pred'}{\embed{\typ_x}}{\smtenv}{\axioms'}
	&
	\tologicshort{\env}{e}{\tfun{x}{\typ_x}{\typ}}{\pred}{\tsmtfun{\embed{\typ_x}}{\embed{\typ}}}{\smtenv}{\axioms}
	& 
	\hastype{\env}{e}{{\typ_x}\rightarrow{\typ}}
}{
	\tologicshort{\env}{e\ e'}{\typ}{\smtappname{\embed{\typ_x}}{\embed{\typ}}\ {\pred}\ {\pred'}}{\embed{\typ}}{\smtenv}{\andaxioms{\axioms}{\axioms'}}
}[\lgapp]
$$

$$
\inference{
	\tologicshort{\env}{e}{\tbool}{\pred}{\tbool}{\smtenv}{\axioms} & 
	\tologicshort{\env}{e_i\subst{x}{e}}{\typ}{\pred_i}{\embed{\typ}}{\smtenv}{\axioms_i}
}{
	\tologicshorttwolines{\env}{\ecaseexp{x}{e}{\etrue \rightarrow e_1; \efalse \rightarrow e_2}}{\typ}
	 {\eif{\pred}{\pred_1}{\pred_2}}{\embed{\typ}}{\smtenv}{\andaxioms{\axioms}{\axioms_i}}
}[\lgcaseBool]
$$

$$
\inference{
	\tologicshort{\env}{e}{\typ_e}{\pred}{\embed{\typ_e}}{\smtenv}{\axioms}\\
	\tologicshort{\env}{e_i\subst{\overline{y_i}}{\overline{\selector{\dc_i}{}\ x}}\subst{x}{e}}{\typ}{\pred_i}{\embed{\typ}}{\smtenv}{\axioms_i}
}{
	\tologicshorttwolines{\env}{\ecase{x}{e}{\dc_i}{\overline{y_i}}{e_i}}{\typ}
	 {\eif{\smtappname{}{}\ \checkdc{\dc_1}\ \pred}{\pred_1}{\ldots} \ \mathtt{else}\ \pred_n}{\embed{\typ}}{\smtenv}
	 {\andaxioms{\axioms}{\axioms_i}}
}[\lgcase]
$$
\caption{\textbf{Transforming \corelan terms into \smtlan.}}
\label{fig:defunc}
\end{figure}

The judgment
\tologicshort{\env}{e}{\typ}{\pred}{\sort}{\smtenv}{\axioms}
states that a $\corelan$ term $e$ is transformed,
under an environment $\env$, into a
$\smtlan$ term $\pred$.
The transformation rules are summarized in Figure~\ref{fig:defunc}.

\mypara{Embedding Types}
We embed \corelan types into \smtlan sorts as:
$$
\begin{array}{rclcrcl}
\embed{\tint}                       & \defeq &  \tint &  &
\embed{T}                           & \defeq &  \tuniv \\
\embed{\tbool}                      & \defeq &  \tbool & &
\embed{\tfun{x}{\typ_x}{\typ}} & \defeq & \tsmtfun{\embed{\typ_x}}{\embed{\typ}}
\end{array}
$$
%%%The embedding extends to typing environments:
%%%% by embedding the types of the environment
%%%$$
%%%\embedsort{\{\tbind{x_1}{\typ_1}, \dots, \tbind{x_n}{\typ_n}\}}
%%%  \defeq
%%%  \{\tbind{x_1}{\embed{\typ_1}}, \dots, \tbind{x_n}{\embed{\typ_n}}
%%%  \}
%%%$$

\mypara{Embedding Constants}
Elements shared on both \corelan and \smtlan
translate to themselves.
These elements include
booleans (\lgbool),
integers (\lgint),
variables (\lgvar),
binary (\lgbinGEN)
and unary (\lgun)
operators.
SMT solvers do not support currying,
and so in \smtlan, all function symbols
must be fully applied.
Thus, we assume that all applications
to primitive constants and data
constructors are \emph{saturated},
%% NV eta converted
\ie fully applied, \eg by converting
source level terms like @(+ 1)@ to
@(\z -> z + 1)@.
%

%%% Thus, to translate \corelan's partially applied operators,
%%% we define an uninterpreted function
%%% $$
%%% \tbind{\smtvar{c}}{\embed{\constty{c}}}
%%% $$
%%% for every functional constant $c$ in \corelan.
%%% %
%%% For example, $+ 1$ will be translated to application of $\smtvar{+}$ to $1$, while
%%% $1+2$ will be translated to the identical $1+2$.

%%\spara{Lambda Lifting}
%%%
%%Since \smtlan does not support $\lambda$-functions.
%%the translation lifts function to axiomatized variables.
%%%
%%Rule~\lgfun
%%translates the term $\efun{x}{\typ}{e}$ to
%%a fresh variable $f$ that satisfies two axioms:
%%(1). $\beta$-reduction,
%%that is $f$ applied to $x$ is equal to $e$, and
%%(2). extentionality,
%%that is for every other function $g$ and argument $x$,
%%if $f$ applied to $x$ is equal to $g$ applied to $x$,
%%then $f = g$.

\mypara{Embedding Functions}
As \smtlan is a first-order logic, we
embed $\lambda$-abstraction and application
using the uninterpreted functions
\smtlamname{}{} and \smtappname{}{}.
We embed $\lambda$-abstractions
using $\smtlamname{}{}$ as shown in rule~\lgfun.
The term $\efun{x}{}{e}$ of type
${\typ_x \rightarrow \typ}$ is transformed
to
${\smtlamname{\sort_x}{\sort}\ x\ \pred}$
of sort
${\tsmtfun{\sort_x}{\sort}}$, where
$\sort_x$ and $\sort$ are respectively
$\embed{\typ_x}$ and $\embed{\typ}$,
${\smtlamname{\sort_x}{\sort}}$
is a special uninterpreted function
of sort
${\sort_x \rightarrow \sort\rightarrow\tsmtfun{\sort_x}{\sort}}$,
and
$x$ of sort $\sort_x$ and $r$ of sort $\sort$ are
the embedding of the binder and body, respectively.
As $\smtlamname{}{}$ is just an SMT-function,
it \emph{does not} create a binding for $x$.
Instead, the binder $x$ is renamed to
a \emph{fresh} name pre-declared in
the SMT environment.

\mypara{Embedding Applications}
Dually, we embed applications via
defunctionalization~\citep{Reynolds72}
using an uninterpreted \emph{apply}
function
$\smtappname{}{}$ as shown in rule~\lgapp.
The term ${e\ e'}$, where $e$ and $e'$ have
types ${\typ_x \rightarrow \typ}$ and $\typ_x$,
is transformed to
${\tbind{\smtappname{\sort_x}{\sort}\ \pred\ \pred'}{\sort}}$
where
$\sort$ and $\sort_x$ are respectively $\embed{\typ}$ and $\embed{\typ_x}$,
the
${\smtappname{\sort_x}{\sort}}$
is a special uninterpreted function of sort
${\tsmtfun{\sort_x}{\sort} \rightarrow \sort_x \rightarrow \sort}$,
and
$\pred$ and $\pred'$ are the respective translations of $e$ and $e'$.

\mypara{Embedding Data Types}
Rule~\lgdc translates each data constructor to a
predefined \smtlan constant ${\smtvar{\dc}}$ of
sort ${\embed{\constty{\dc}}}$.
Let $\dc_i$ be a non-boolean data constructor such that
$$
\constty{\dc_i} \defeq \typ_{i,1} \rightarrow \dots \rightarrow \typ_{i,n} \rightarrow \typ
$$
Then the \emph{check function}
${\checkdc{{\dc_i}}}$ has the sort
$\tsmtfun{\embed{\typ}}{\tbool}$,
and the \emph{select function}
${\selector{\dc}{i,j}}$ has the sort
$\tsmtfun{\embed{\typ}}{\embed{\typ_{i,j}}}$.
Rule~\lgcase translates case-expressions
of \corelan into nested $\mathtt{if}$
terms in \smtlan, by using the check
functions in the guards, and the
select functions for the binders
of each case.
%
%\mypara{Reflecting DataTypes}
%
% The above approach  makes it straightforward
% to reflect functions over datatypes into \smtlan.
%
For example, following the above, the body of the list append function
%
%%% reflect (++) :: xs:[Int] -> ys:[Int] -> [Int]
\begin{code}
  []     ++ ys = ys
  (x:xs) ++ ys = x : (xs ++ ys)
\end{code}
is reflected into the \smtlan refinement:
$$
\ite{\mathtt{isNil}\ \mathit{xs}}
    {\mathit{ys}}
    {\mathtt{sel1}\ \mathit{xs}\
       \dcons\
       (\mathtt{sel2}\ \mathit{xs} \ \mathtt{++}\  \mathit{ys})}
$$
We favor selectors to the axiomatic translation of
HALO~\citep{HALO} and \fstar~\cite{fstar} to avoid
universally quantified formulas and the resulting
instantiation unpredictability.

%% $$
%% \tbind{\checkdc{\dc}}{\embed{\typ \rightarrow \tbool}}
%% \ \text{with}\ \constty{\dc} = \typ_1 \rightarrow \dots \rightarrow \typ_n\rightarrow\typ
%% $$
%% and the field selector is used to substitute the data constructor quantified variables $\overline{y_i}$:
%% eg. if \dc is [] then i == 0
%%     if \dc is (:) :: a -> [a] -> [a] then
%%         \dc_1 = head :: [a] -> a
%%         \dc_2 = tail :: [a] -> [a]
%% $$
%% \tbind
      %% {\embed{\typ \rightarrow \typ_i}}
%% \ \text{with}\ \constty{\dc} = \typ_1 \rightarrow \dots \rightarrow \typ_n\rightarrow\typ, i \leq n
%% $$
%% %
%% For example, the body of the @length@ function from~\S~\ref{sec:examples}
%% translates to the condition $\eif{\isN\ xs}{0}{1+\texttt{length} (\etail\ xs)}$,
%% as $\etail \defeq \selector{\dcons}{2}$.

\subsection{Correctness of Translation}

Informally, the translation relation $\tologicshort{\env}{e}{}{\pred}{}{}{}$
is correct in the sense that if $e$ is a terminating boolean expression
then $e$ reduces to \etrue \textit{iff} $\pred$ is SMT-satisfiable
by a model that respects $\beta$-equivalence.

\NV{below we use substitution in lambda s which is not formally defined}
\begin{definition}[$\beta$-Model]\label{def:beta-model}
A $\beta-$model $\bmodel$ is an extension of a model $\sigma$
where $\smtlamname{}{}$ and $\smtappname{}{}$
satisfy the axioms of $\beta$-equivalence:
$$
\begin{array}{rcl}
\forall x\ y\ e. \smtlamname{}{}\ x\ e
  & = & \smtlamname{}{}\ y\ (e\subst{x}{y}) \\
\forall x\ e_x\ e. (\smtappname{}{}\ (\smtlamname{}{}\ x\ e)\ e_x
  & = &  e\subst{x}{e_x}
\end{array}
$$
\end{definition}

\mypara{Semantics Preservation}
We define the translation of a \corelan term
into \smtlan under the empty environment as
${\embed{e} \defeq \pred}$
if ${\tologicshort{\emptyset}{\refa}{}{\pred}{}{}{}}$.
A \emph{lifted substitution}
$\theta^\perp$ is a set of models $\sigma$
where each ``bottom'' in the substitution
$\theta$ is mapped to an arbitrary logical
value of the respective sort~\citep{Vazou14}.
We connect the semantics of \corelan and translated
\smtlan via the following theorems:
% terms can connect evaluation of boolean
% \corelan expression to \smtlan predicates.

\begin{theorem}\label{thm:embedding-general}
If ${\tologicshort{\env}{\refa}{}{\pred}{}{}{}}$,
then for every ${\sub\in\interp{\env}}$
and every ${\sigma\in {\sub^\perp}}$,
if $\evalsto{\applysub{\sub^\perp}{\refa}}{v}$
then $\sigma^\beta \models \pred = \embed{v}$.
\end{theorem}

% For Boolean expressions we specialize the above to

\begin{corollary}\label{thm:embedding}
If ${\hastype{\env}{\refa}{\tbool}}$, $e$ reduces to a value and
${\tologicshort{\env}{\refa}{\tbool}{\pred}{\tbool}{\smtenv}{\axioms}}$,
then for every ${\sub\in\interp{\env}}$
and every ${\sigma\in {\sub^\perp}}$,
$\evalsto{\applysub{\sub^\perp}{\refa}}{\etrue}$ iff
$\sigma^\beta \models \pred$.
\end{corollary}

\subsection{Decidable Type Checking}
\begin{figure}[t!]
\centering
$$
\begin{array}{rrcl}
\emphbf{Refined Types} \quad
  & \typ
  & ::=   & \tref{v}{\btyp^{[\tlabel]}}{\reft} \spmid \tfun{x}{\typ}{\typ}
\\[0.10in]
\end{array}
$$
\emphbf{Well Formedness}\hfill{\fbox{\aiswellformed{\env}{\typ}}}\\
$$
\inference{
  \ahastype{\env,\tbind{v}{\btyp}}{\refa}{\tbool^{\tlabel}}
}{
  \aiswellformed{\env}{\tref{v}{\btyp}{\refa}}
}[\rwbase]
$$
\emphbf{Subtyping}\hfill{\fbox{\aissubtype{\env}{\typ}{\typ'}}}\\
$$
\inference{
\env' \defeq \env,\tbind{v}{\{\btyp^\tlabel | \refa\}} &
\tologicshort{\env'}{\refa'}{\tbool}{\pred'}{}{}{} &
\smtvalid{\vcond{\env'}{\pred'}}
}{
 \aissubtype{\env}{\tref{v}{\btyp}{\refa}}{\tref{v}{\btyp}{\refa'}}
}[\rsubbase]
$$
%%%% %\NV{REVERT TO OLD DEFINITIONS, what is e'?}
%%%% $$
%%%% \inference{
%%%% \tologicshort{\env'}{\refa_1}{\tbool}{\pred_1}{\tbool}{\smtenv_1}{\axioms_1} &
%%%% \tologicshort{\env'}{\refa_2}{\tbool}{\pred_2}{\tbool}{\smtenv_1}{\axioms_1} \\
%%%% % \isvalid{\env,\tbind{v}{\btyp}}{\refa_1}{\refa_2}
%%%% \env' \defeq \env,\tbind{v}{\btyp^\tlabel} &
%%%% % \tologicshort{\env'}{\refa'}{\tbool}{\pred'}{\tbool}{\smtenv'}{\axioms'} &
%%%% \smtvalid{\vcond{\env'}{\pred_1 \Rightarrow \pred_2}}
%%%% %
%%%% }{
  %%%% \aissubtype{\env}{\tref{v}{\btyp}{\refa_1}}{\tref{v}{\btyp}{\refa_2}}
%%%% }[\rsubbase]
%%%% $$
%%% \emphbf{Implication}\hfill{\isvalid{\env}{\refa_1}{\refa_2}}\\
%%% $$
%%% \inference{
  %%% \tologicshort{\env}{\refa_1}{\tbool}{\pred_1}{\tbool}{\smtenv_1}{\axioms_1} &
  %%% \tologicshort{\env}{\refa_2}{\tbool}{\pred_2}{\tbool}{\smtenv_2}{\axioms_i} \\
  %%% \text{is SMT-valid}\ (\embedexpr{\env} \Rightarrow \pred_1 \Rightarrow \pred_2)
%%% }{
  %%% \isvalid{\env}{\refa_1}{\refa_2}
%%% }
%%% $$
%%% \emphbf{Typing}\hfill{\ahastype{\env}{\prog}{\typ}}\\
\caption{\textbf{Algorithmic Typing (other rules in Figs~\ref{fig:syntax} and \ref{fig:typing}.)}}
\label{fig:modifications}
\end{figure}

Figure~\ref{fig:modifications} summarizes the modifications required
to obtain decidable type checking.
Namely, basic types are extended with labels that track termination
and subtyping is checked via an SMT solver.

\mypara{Termination}
Under arbitrary beta-reduction semantics
(which includes lazy evaluation), soundness
of refinement type checking requires checking
termination, for two reasons:
(1)~to ensure that refinements cannot diverge, and
(2)~to account for the environment during subtyping~\citep{Vazou14}.
We use \tlabel to mark provably terminating
computations, and extend the rules to use
refinements to ensure that if
${\ahastype{\env}{e}{\tref{v}{\btyp^\tlabel}{r}}}$,
then $e$ terminates~\citep{Vazou14}.
%
%% Here we assume termination is checked by an oracle,
%% but we can use refinement types themselves to prove
%% correctness of the termination labeling

\mypara{Verification Conditions}
The \emph{verification condition} (VC)
${\vcond{\env}{\pred}}$
is \emph{valid} only if the set of values
described by $\env$, is subsumed by
the set of values described by $\pred$.
$\env$ is embedded into logic by conjoining
(the embeddings of) the refinements of
provably terminating binders~\cite{Vazou14}:
%
%% We only trust refinements of terminating
%% expressions, as every diverging expression
%% can be unsoundly refined \efalse.
%% $$
%% \embed{\env} \defeq
  %% \bigwedge\{ p \mid \tbind{x}{\tref{v}{\btyp^{\tlabel}}{e}} \in \env
   %% \land \tologicshort{\env}{e\subst{v}{x}}{\btyp}{p}{\embed{\btyp}}{\smtenv}{\axioms}
   %% \}
%% $$
\begin{align*}
\embed{\env} \defeq & \bigwedge_{x \in \env} \embed{\env, x} \\
\intertext{where we embed each binder as}
\embed{\env, x} \defeq & \begin{cases}
                           \pred  & \text{if } \env(x)=\tref{v}{\btyp^{\tlabel}}{e},\
                                    \tologicshort{\env}{e\subst{v}{x}}{\btyp}{\pred}{\embed{\btyp}}{\smtenv}{\axioms} \\
                           \etrue & \text{otherwise}.
                         \end{cases}
\end{align*}

%We use the embedding of environment to decidably check subtyping.
%As defined in Figure~\ref{fig:modifications},
%\tref{v}{\btyp}{\refa_1} is subtype of \tref{v}{\btyp}{\refa_1}
%under the environment \env, when
%$\refa_i$ transforms to $\pred_i$ with axioms $\axioms_i$
%and assuming $\embedexpr{\env}$ and the axioms $\axioms_i$
%$\pred_i$ implies $\pred_2$.

\mypara{Subtyping via SMT Validity}
We make subtyping, and hence, typing decidable,
by replacing the denotational base subtyping
rule $\rsubbase$ with a conservative,
algorithmic version that uses an SMT
solver to check the validity of the subtyping VC.
We use Corollary~\ref{thm:embedding} to prove
soundness of subtyping. 
\begin{lemma}\label{lem:subtyping} %[Conservative Subtyping]
If {\aissubtype{\env}{\tref{v}{\btyp}{e_1}}{\tref{v}{\btyp}{e_2}}}
then {\issubtype{\env}{\tref{v}{\btyp}{e_1}}{\tref{v}{\btyp}{e_2}}}.
\end{lemma}

\mypara{Soundness of \smtlan}
Lemma~\ref{lem:subtyping} directly implies the soundness of \smtlan.
\begin{theorem}[Soundness of \smtlan]\label{thm:soundness-smt}
If \ahastype{\env}{e}{\typ} then \hastype{\env}{e}{\typ}.
\end{theorem}

\begin{comment}
\begin{proof}
By rule \rsubbase, we need to show that
$\forall \sub\in\interp{\env}.
  \interp{\applysub{\sub}{\tref{v}{\btyp}{\refa_1}}}
  \subseteq
  \interp{\applysub{\sub}{\tref{v}{\btyp}{\refa_2}}}$.
%
We fix a $\sub\in\interp{\env}$.
and get that forall bindings
$(\tbind{x_i}{\tref{v}{\btyp^{\downarrow}}{\refa_i}}) \in \env$,
$\evalsto{\applysub{\sub}{e_i\subst{v}{x_i}}}{\etrue}$.

Then need to show that for each $e$,
if $e \in \interp{\applysub{\sub}{\tref{v}{\btyp}{\refa_1}}}$,
then $e \in \interp{\applysub{\sub}{\tref{v}{\btyp}{\refa_2}}}$.

If $e$ diverges then the statement trivially holds.
Assume $\evalsto{e}{w}$.
We need to show that
if $\evalsto{\applysub{\sub}{e_1\subst{v}{w}}}{\etrue}$
then $\evalsto{\applysub{\sub}{e_2\subst{v}{w}}}{\etrue}$.

Let \vsub the lifted substitution that satisfies the above.
Then  by Lemma~\ref{thm:embedding}
for each model $\bmodel \in \interp{\vsub}$,
$\bmodel\models\pred_i$, and $\bmodel\models q_1$
for
$\tologicshort{\env}{e_i\subst{v}{x_i}}{\btyp}{\pred_i}{\embed{\btyp}}{\smtenv_i}{\axioms_i}$
$\tologicshort{\env}{e_i\subst{v}{w}}{\btyp}{q_i}{\embed{\btyp}}{\smtenv_i}{\beta_i}$.
%
Since \aissubtype{\env}{\tref{v}{\btyp}{e_1}}{\tref{v}{\btyp}{e_2}} we get
$$
\bigwedge_i \pred_i
\Rightarrow q_1 \Rightarrow q_2
$$
thus $\bmodel\models q_2$.
%
By Theorem~\ref{thm:embedding} we get $\evalsto{\applysub{\sub}{\refa_2\subst{v}{w}}}{\etrue}$.
\end{proof}
\end{comment}

\section{Reasoning About Lambdas}\label{sec:lambdas}

Though \smtlan, as presented so far, is sound and decidable,
it is \emph{imprecise}: our encoding of $\lambda$-abstractions
and applications via uninterpreted functions makes it impossible
to prove theorems that require $\alpha$- and $\beta$-equivalence,
or extensional equality. Next, we show how to address the former
by strengthening the VCs with equalities~\S~\ref{subsec:equivalences},
and the latter by introducing a combinator for safely asserting
extensional equality~\S~\ref{subsec:extensionality}.
In the rest of this section, 
for clarity we omit $\smtappname{}{}$ when it is clear from the context.

\subsection{Equivalence}\label{subsec:equivalences}

As soundness relies on satisfiability under a
\bmodel  (see Definition~\ref{def:beta-model}),
we can safely \emph{instantiate} the axioms of
$\alpha$- and $\beta$-equivalence on any set of
terms of our choosing and still preserve soundness
(Theorem~\ref{thm:soundness-smt}).
That is, instead of checking the validity
of a VC
${p \Rightarrow q}$,
we check the validity of a \emph{strengthened VC},
${a \Rightarrow p \Rightarrow q}$,
where $a$ is a (finite) conjunction
of \emph{equivalence instances}
derived from $p$ and $q$,
as discussed below.

% Since it is unclear how to reify this axiomatization
% while preserving decidability, we choose (once again)
% to syntactically instantiate the axioms of $\alpha$-,
% $\beta$- and normal form equivalence on the relevant
% candidates.

\mypara{Representation Invariant}
The lambda binders,
for each SMT sort, are drawn from a
pool of names $x_i$ where the index
$i=1,2,\ldots$.
When representing
$\lambda$ terms we enforce
a \emph{normalization invariant}
that for each lambda term
$\slam{x_i}{e}$, the index $i$
% $x_i$
is greater than any lambda
argument appearing in $e$.

\mypara{$\alpha$-instances}
For each syntactic term
${\slam{x_i}{e}}$, and $\lambda$-binder
$x_j$ such that $i < j$ appearing in the VC,
we generate an \emph{$\alpha$-equivalence instance predicate}
(or \emph{$\alpha$-instance}):
$$\slam{x_i}{e} = \slam{x_j}{e \subst{x_i}{x_j}}$$

% , x_i < x_j \leq \maxlamarg$$
%\leq \maxlamarg$
% that follow the ordering
% $i < j \Leftrightarrow x_i < x_j$.
%
% The number of distinct valid lambda
% arguments is determined by the
% maximum number of $\lambda$s
% syntactically in program's refinements.
% In the implementation, we bound
% this number by $\maxlamarg = 7$.
% %
% In practice we only encountered
% 4 nested lambdas when verifying
% monadic left identity of the
% reader monad.

The conjunction of $\alpha$-instances
can be more precise than De Bruijn
representation, as they let the SMT
solver deduce more equalities via
congruence.
For example, consider the VC needed
to prove the applicative laws for @Reader@:
\begin{align*}
d & = \slam{x_1}{(\sapp{x}{x_1})} \\
  & \Rightarrow \slam{x_2}{(\sapp{(\slam{x_1}{(\sapp{x}{x_1})})}{x_2})}
              = \slam{x_1}{(\sapp{d}{x_1})}
\end{align*}
The $\alpha$ instance
${\slam{x_1}{(\sapp{d}{x_1})} = \slam{x_2}{(\sapp{d}{x_2})}}$
derived from the VC's hypothesis,
combined with congruence immediately
yields the VC's consequence.

\mypara{$\beta$-instances}
For each syntactic term $\smapp{(\slam{x}{e})}{e_x}$,
with $e_x$ not containing any $\lambda$-abstractions,
appearing in the VC,
we generate an \emph{$\beta$-equivalence instance predicate}
(or \emph{$\beta$-instance}):
$$
\smapp{(\slam{x_i}{e})}{e_x} = \SUBST{e}{x_i}{e_x},
  \mbox{ s.t. } e_x \mbox{ is $\lambda$-free}
$$
We require the $\lambda$-free restriction as
a simple way to enforce that the reduced
term ${\SUBST{e}{x_i}{e'}}$ enjoys the
representation invariant.
%
%%\NV{This restriction is not implemented, but neither has a normalization function}
%%\RJ{Then why are we discussing it?}

For example, consider the following VC
needed to prove that the bind operator for
lists satisfies the monadic associativity law.
$$(\sapp{f}{x} \ebind g) = \smapp{(\slam{y}{(\sapp{f}{y} \ebind g)})}{x}$$
The right-hand side of the above VC generates
a $\beta$-instance that corresponds directly
to the equality, allowing the SMT solver to
prove the (strengthened) VC.

%% immediately yields
%% $\beta$-equivalence is required in our benchmarks to for example prove
%% monadic associativity for lists.
%% %
%% An intermediate step in the proof is to verify that
%%
%% Assuming that @h x = f x >>= g@ the above
%% translates to the logical query
%% $$
%% \sapp{h}{x} = \sapp{(\slam{x_1}{(\sapp{h}{x_1})})}{x}
%% $$
%%
%% The right hand side of the query will fire a
%% $\beta$-equivalence instantiation
%% and the assumption
%% $
%% \sapp{(\slam{x_1}{(\sapp{h}{x_1})})}{x} = \sapp{h}{x}
%% $
%% will be added in the environment, allowing SMT to prove the equivalence.

\mypara{Normalization}
The combination of $\alpha$- and $\beta$-instances
is often required to discharge proof obligations.
For example, when proving that the bind operator
for the @Reader@ monad is associative, we need
to prove the VC:
$$\slam{x_2}{(\slam{x_1}{w})} =
  \slam{x_3}{(\smapp{(\slam{x_2}{(\slam{x_1}{w})})}{w})}$$
The SMT solver proves the VC via the equalities
corresponding to an $\alpha$ and then $\beta$-instance:
\begin{align*}
\slam{x_2}{(\slam{x_1}{w})}
  \ =_{\alpha}\ & \slam{x_3}{(\slam{x_1}{w})} \\
  \ =_{\beta}\ & \slam{x_3}{(\smapp{(\slam{x_2}{(\slam{x_1}{w})})}{w})}
\end{align*}

%%% of monadic
%%% associativity for the reader monad requires to prove a $\lambda$ equality
%%% simplified to
%%% %
%%% \begin{code}
  %%% w:a -> {(\x y -> w) = (\x -> (\z y -> w) w)}
%%% \end{code}
%%% %
%%% The proof for the above code is @trivial@ as
%%% in the logic, the $\lambda$ arguments are renamed to
%%% @(\x2 x1 -> w) = (\x3 -> (\x2 x1 -> w) w)@.
%%% Due to $\alpha$- equivalence the
%%% right hand side is equal to
%%% @\x3 x1 -> w@.
%%% Due to $\beta$-equivalence we get
%%% @((\x2 x1 -> w) w) = \x1 -> w@,
%%% by which and due to congruence axiom
%%% we get the desired equality,
%%% \begin{code}
  %%% \x y -> w               -- lhs
  %%% \x2 x1 -> w             -- representation
  %%% \x3 x1 -> w             -- alpha
  %%% \x3 ->((\x2 x1 -> w) w) -- beta
  %%% \x -> ((\z y -> w) w)   -- rhs
%%% \end{code}

\subsection{Extensionality} \label{subsec:extensionality}

Often, we need to prove that two
functions are equal, given the
definitions of reflected binders.
For example, consider
\begin{code}
 reflect id
 id x = x
\end{code}
\toolname accepts the proof that
@id x = x@ for all @x@:
\begin{code}
 id_x_eq_x :: x:a -> {id x = x}
 id_x_eq_x = \x -> #id# x =. x ** QED
\end{code}
as ``calling'' @id@ unfolds its definition,
completing the proof.
However, consider this $\eta$-expanded variant of
the above proposition:
\begin{code}
 type Id_eq_id = {(\x -> id x) = (\y -> y)}
\end{code}
\toolname \emph{rejects} the proof:
\begin{code}
 fails :: Id_eq_id
 fails =  (\x -> #id# x) =. (\y -> y) ** QED
\end{code}
The invocation of @id@ unfolds the
definition, but the resulting equality
refinement @{id x = x}@ is \emph{trapped}
under the $\lambda$-abstraction.
That is, the equality is absent from the
typing environment at the \emph{top} level,
where the left-hand side term is compared to @\y -> y@.
%
% NV LHS reads like LiquidHaskell
Note that the above equality requires
the definition of @id@ and hence is
outside the scope of purely the
$\alpha$- and $\beta$-instances.

\newcommand\eqfun{\ensuremath{\texttt{=}\forall}}

\mypara{An Exensionality Operator}
To allow function equality via
extensionality, we provide the
user with a (family of)
\emph{function comparison operator(s)}
that transform an \emph{explanation} @p@
which is a proof that @f x = g x@ for every
argument @x@, into a proof that @f = g@.
\begin{mcode}
  =* :: f:(a -> b) -> g:(a -> b)
     -> exp:(x:a -> {f x = g x})
     -> {f = g}
\end{mcode}
Of course, @=*@ cannot be implemented;
its type is \emph{assumed}. We can use
@=*@ to prove @Id_eq_id@ by providing
a suitable explanation:
\begin{mcode}
pf_id_id :: Id_eq_id
pf_id_id = (\y -> y) =* (\x -> id x) $\because$ expl ** QED
  where
    expl = (\x -> #id# x =. x ** QED)
\end{mcode}
%
%   =* (\x -> id x) $\because$ exp
%   where
%    exp   :: x:a -> {(\x -> id x) x = (\x -> x) x}
%    exp x = #id# x =. x ** QED
%
The explanation is
% the proof passed as
the second argument to $\because$ which has
the following type that syntactically fires $\beta$-instances:
\begin{code}
  x:a -> {(\x -> id x) x = ((\x -> x) x}
\end{code}

\section{Evaluation}\label{sec:evaluation}

% The source counts were generated by:
% $ cd benchmarks/haskell16/pos/
% $ find . -type f -name '*.hs' -exec sed -i '' s/\{-@/\{-#LH/ {} +
% $ sloccount

\begin{figure}[t!]
\begin{center}
\begin{tabular}{lllr}
\toprule
  \multicolumn{3}{l}{\textbf{CATEGORY}}              & \textbf{LOC} \\
\toprule
  \textbf{I.} & \multicolumn{3}{l}{\textbf{Arithmetic}} \\[0.05in]
   & Fibonacci      & \S~\ref{sec:overview}          &  48 \\ % Overview.hs
   & Ackermann      & \citep{ackermann}
                    , Fig.~\ref{fig:ackermann}       & 280 \\ % Ackermann.hs

  \midrule

  \textbf{II.} & \multicolumn{3}{l}{\textbf{Algebraic Data Types}} \\[0.05in]

  & Fold Universal & \citep{agdaequational}          & 105 \\ % FoldrUniversal.hs
  & Fold Fusion    & \citep{agdaequational}          &     \\

  \midrule

  \textbf{III.} & \textbf{Typeclasses} & Fig~\ref{fig:laws} & \\[0.05in]
  & Monoid         & \tPeano, \tMaybe, \tList        & 189 \\ % Monoid*.hs
  & Functor        & \tMaybe, \tList, \tId, \tReader & 296 \\ % Functor*.hs - FunctorReader.hs
  & Applicative    & \tMaybe, \tList, \tId, \tReader & 578 \\ % Applicative*.hs
  & Monad          & \tMaybe, \tList, \tId, \tReader & 435 \\ % Monad*.hs

  \midrule

  \textbf{IV.} & \multicolumn{3}{l}{\textbf{Functional Correctness}} \\[0.05in]
  & SAT Solver     & \citep{Zombie}                  & 133 \\ % Solver.hs
  & Unification    & \citep{Sjoberg2015}             & 200 \\ % Unification

  \midrule

  \multicolumn{3}{l}{\textbf{TOTAL}}                 & 2264 \\
\bottomrule
\end{tabular}
\end{center}
\caption{\textbf{Summary of Case Studies}}
\label{fig:eval-summary}
\end{figure}

We have implemented refinement reflection
in \toolname. 
In this section, we evaluate our approach
by using \toolname to verify a variety of
deep specifications of Haskell functions
drawn from the literature and categorized
in Figure~\ref{fig:eval-summary},
totalling about 2500 lines of specifications
and proofs.
Next, we detail each of the four classes of
specifications, illustrate how they were
verified using refinement reflection, and
discuss the strengths and weaknesses of
our approach.
\emph{All} of these proofs require refinement
reflection, \ie are beyond the scope of shallow
refinement typing.

\mypara{Proof Strategies.}
Our proofs use three building blocks, that are seamlessly
connected via refinement typing:
\begin{itemize}
  \item \emphbf{Un/folding}
     definitions of a function @f@ at
     arguments @e1...en@, which due
     to refinement reflection, happens
     whenever the term @f e1 ... en@
     appears in a proof.
     For exposition, we render the function
     whose un/folding is relevant as @#f#@;

  \item \emphbf{Lemma Application}
     which is carried out by using
     the ``because'' combinator
     ($\because$) to instantiate
     some fact at some inputs;

  \item \emphbf{SMT Reasoning}
     in particular, \emph{arithmetic},
     \emph{ordering} and \emph{congruence closure}
     which kicks in automatically (and predictably!),
     allowing us to simplify proofs by not
     having to specify, \eg which subterms
     to rewrite.
\end{itemize}

\subsection{Arithmetic Properties} \label{subsec:arith} \label{subsec:ackermann}

The first category of theorems pertains to the textbook
Fibonacci and Ackermann functions.
The former were shown in \S~\ref{sec:overview}.
The latter are summarized in Figure~\ref{fig:ackermann},
which shows two alternative definitions for the
Ackermann function.
We proved equivalence of the definition (Prop 1)
and various arithmetic relations between
them (Prop 2 --- 13), by mechanizing the
proofs from~\cite{ackermann}.

\begin{figure}[t!]
\textbf{Ackermann's Function}
\[
\begin{array}{lr}
\ack{n}{x} \defeq
 \left\{
\begin{array}{l}
\setlength\arraycolsep{0pt}
\begin{array}{ll}
      x+2 &\quad \text{, if}\ n=0 \\
      2   &\quad \text{, if}\ x=0 \\
\end{array} \\
\ack{n-1}{\ack{n}{x-1}}
\end{array}
\right.
&
\iack{h}{n}{x} \defeq
 \left\{
\begin{array}{l}
      x \quad \text{, if}\ h=0 \\
      \ack{n}{\iack{h-1}{n}{x}}
\end{array}
\right.
\end{array}
 \]

\textbf{Properties}
$$
\begin{array}{lrcrcl}
1.&                &&            \ack{n+1}{x}   &=& \iack{x}{n}{2}\\
2.&                &&            x + 1          &<& \ack{n}{x}\\
3.&                &&            \ack{n}{x}     &<& \ack{n}{x+1}\\
4.& x < y          &\Rightarrow& \ack{n}{x}     &<& \ack{n}{y}\\
5.& 0 < x          &\Rightarrow& \ack{n}{x}     &<& \ack{n+1}{x}\\
6.& 0 < x, n < m   &\Rightarrow& \ack{n}{x}     &<& \ack{m}{x}\\
7.&                &&            \iack{h}{n}{x} &<& \iack{h+1}{n}{x}\\
8.&                &&            \iack{h}{n}{x} &<& \iack{h}{n}{x+1}\\
9.& x<y            &\Rightarrow& \iack{h}{n}{x} &<& \iack{h}{n}{y}\\
10.&               &\Rightarrow& \iack{h}{n}{x} &<& \iack{h}{n+1}{x}\\
11.& 0<n, l-2 < x  &\Rightarrow& x + l          &<& \ack{n}{x}\\
12.& 0<n, l-2 < x  &\Rightarrow& \iack{l}{n}{x} &<& \ack{n+1}{x}\\
13.&               &&            \iack{x}{n}{y} &<& \ack{n+1}{x+y}\\
\end{array}
$$
\caption{\textbf{Ackermann Properties~\citep{ackermann},
$\forall n, m, x, y, h, l \geq 0$}}
\label{fig:ackermann}
\end{figure}

\mypara{Monotonicity}
Prop 3. shows that \ack{n}{x} is increasing on $x$.
We derived Prop 4. by applying @fMono@ theorem
from \S~\ref{sec:examples} with input function
the partially applied Ackermann Function
$\ack{n}{\star}$.
Similarly, we derived the monotonicity Prop 9. by
applying @fMono@ to the locally increasing Prop. 8
and $\iack{h}{n}{\star}$.
Prop 5. proves that \ack{n}{x} is increasing
on the \emph{first} argument $n$.
As @fMono@ applies to the \emph{last} argument
of a function, we cannot directly use it to
derive Prop 6.
Instead, we define a variant @fMono2@ that works
on the first argument of a binary function, and
use it to derive Prop 6.

%%\mypara{Existentials}
%%%
%%Properties 11. and 12. are described in~\citep{ackermann}
%%to hold for almost every $x$.
%%%
%%That is, Property 11. is described as
%%$\exists x_0. x_0 < x \Rightarrow x + l < \ack{n}{x}$.
%%%
%%Refinement types cannot express existentials,
%%thus expressing the above property is (currently)
%%not feasible, but the minimum $x$ was easy to retrieve.
%%%
%%Thus, we expressed the above statement by specifying $x_0 = l - 2$.
%%%
%%\NV{``Almost'' means in English for large x.
%%In math it translates to there exist some x0 such that...
%%in LiquidHaskell, and due to lack of existentials
%%I had to find (``retrieve'') an x0 = l-2 above which
%%the property holds For almost see big o notation, where
%%there exists a c such that...to specify these properties
%%in LH you need to give a concrete c But, this paragraph
%%summarizes a lot of internal thinking, so feel free to
%%rephrase. I wanted it in, because when I saw the lemma
%%stating this property holds almost everywhere I though
%%we could not express it, but we can!}

%
\mypara{Constructive Proofs}
In \citep{ackermann} Prop 12. was proved by constructing
an auxiliary \emph{ladder} that counts the number of
(recursive) invocations of the Ackermann function, and
uses this count to bound \iack{h}{n}{x} and \ack{n}{x}.
It turned out to be straightforward and natural
to formalize the proof just by defining the
@ladder@ function in Haskell, reflecting it,
and using it to formalize the algebra from~\citep{ackermann}.

\subsection{Algebraic Data Properties}
\label{subsec:fold}

The second category of properties pertain to
algebraic data types.
%, \eg \emph{folding} over lists.

\mypara{Fold Univerality}
Next, we proved properties of list folding, such as
the following, describing the \emph{universal}
property of right-folds~\citep{agdaequational}:
\begin{code}
foldr_univ
  :: f:(a -> b -> b)
  -> h:([a] -> b)
  -> e:b
  -> ys:[a]
  -> base:{h [] = e }
  -> stp:(x:a ->l:[a]->{h(x:l) = f x (h l)})
  -> {h ys = foldr f e ys}
\end{code}
Our proof @foldr_univ@ differs from the one in Agda,
in two ways.
First, we encode Agda's universal quantification over
@x@ and @l@ in the assumption @stp@ using a function type.
Second, unlike Agda, \toolname
does not support implicit arguments,
so at \emph{uses} of @foldr_univ@
the programmer must explicitly
provide arguments for @base@
and @stp@, as illustrated below.

\mypara{Fold Fusion}
Let us define the usual composition operator:
\begin{code}
reflect . :: (b -> c) -> (a -> b) -> a -> c
f . g     = \x -> f (g x)
\end{code}
We can prove the following @foldr_fusion@ theorem
(that shows operations can be pushed inside a @foldr@),
by applying @foldr_univ@ to explicit @bas@ and @stp@ proofs:
\begin{code}
  foldr_fusion
   :: h:(b -> c)
   -> f:(a -> b -> b)
   -> g:(a -> c -> c)
   -> e:b -> z:[a] -> x:a -> y:b
   -> fuse: {h (f x y) = g x (h y)})
   -> {(h . foldr f e) z = foldr g (h e) z}

  foldr_fusion h f g e ys fuse
    = foldr_univ g (h . foldr f e) (h e) ys
        (fuse_base h f e)
        (fuse_step h f e g fuse)
\end{code}
where @fuse_base@ and @fuse_step@ prove the
base and inductive cases, and for example
@fuse_base@ is a function with type
\begin{code}
fuse_base :: h:(b->c) -> f:(a->b->b) -> e:b
          -> {(h . foldr f e) [] = h e}
\end{code}

\subsection{Typeclass Laws}

\begin{figure}[t!]
\begin{center}
\begin{tabular}{rl}

\toprule

\multicolumn{2}{c}{\textbf{Monoid}} \\
{Left Ident.}  & $\emempty\ x\ \emappend\  \equiv x$  \\
{Right Ident.} & $x\ \emappend\ \emempty \equiv x$  \\
{Associativity}  & $(x\ \emappend\ y)\ \emappend\ z \equiv x\ \emappend\ (y\ \emappend\ z)$ \\

\midrule

\multicolumn{2}{c}{\textbf{Functor}} \\
{Ident.}     & $\efmap\ \eid\ xs \equiv \eid\ xs$ \\
{Distribution} & $\efmap\ (g\ecompose\ h)\ xs \equiv (\efmap\ g\ \ecompose\ \efmap\ h)\ xs$\\

\midrule

\multicolumn{2}{c}{\textbf{Applicative}} \\

{Ident.}      & $\epure \eid \eseq\ v \equiv v$ \\
{Compos.}     & $\epure (\ecompose) \eseq u \eseq v \eseq w \equiv u \eseq (v \eseq w)$ \\
{Homomorph.}  & $\epure\ f\ \eseq\ \epure\ x \equiv \epure\ (f\ x)$\\
{Interchange} & $u\ \eseq\ \epure\ y \equiv \epure\ (\$\ y) \ \eseq \ u$ \\

\midrule
\multicolumn{2}{c}{\textbf{Monad}} \\
{Left Ident.}   & $\ereturn\ a \ebind f \equiv f\ a$ \\
{Right Ident.}  & $m \ebind \ereturn \equiv m$ \\
{Associativity} & $(m\ebind f) \ebind g \equiv m\ebind (\lambda x \rightarrow f\ x \ebind g)$\\
\bottomrule
\end{tabular}
\end{center}
\caption{\textbf{Typeclass Laws verified using \toolname}}
\label{fig:laws}
\end{figure}
We used \toolname to prove the Monoid, Functor,
Applicative and Monad Laws, summarized in
Figure~\ref{fig:laws}, for various user-defined
instances summarized in Figure~\ref{fig:eval-summary}.

%% The purpose of these proofs is to investigate the
%% proving abilities of \libname.
%% For this purpose, we defined the appropriate class
%% operators on user defined lists, instead of using
%% Haskell's predefined class instances.
%% %
%% In the near future, we plan to embed these proofs
%% to check the laws on real Haskell instances,
%% but this requires some engineering from the
%% \liquidHaskell team.

\mypara{Monoid Laws}
A Monoid is a datatype equipped with an associative
binary operator $\emappend$ and an \emph{identity}
element $\emempty$.
We use \toolname to prove that
@Peano@ (with @add@ and @Z@),
@Maybe@ (with a suitable @mappend@ and @Nothing@), and
@List@ (with append @++@ and @[]@) satisfy the monoid laws.
For example, we prove that @++@ (\S~\ref{subsec:list})
is associative by reifying the textbook proof~\cite{HuttonBook}
into a Haskell function, where the induction
corresponds to case-splitting and recurring
on the first argument:
\begin{mcode}
assoc :: xs:[a] -> ys:[a] -> zs:[a] ->
       {(xs ++ ys) ++ zs = xs ++ (ys ++ zs)}

assoc [] ys zs     = ([] #++# ys) ++ zs
                   =. [] #++# (ys ++ zs)
                   ** QED
assoc (x:xs) ys zs = ((x:  xs)#++# ys) ++ zs
                   =. (x: (xs ++ ys))#++# zs
                   =.  x:((xs ++ ys) ++ zs)
                   =.  x: (xs ++ (ys ++ zs))
                       $\because$ assoc xs ys zs
                   =. (x:xs)  #++# (ys ++ zs)
                   ** QED
\end{mcode}

\mypara{Functor Laws}
A type is a functor if it has a function
@fmap@ that satisfies the \emph{identity}
and \emph{distribution} (or fusion) laws
in Figure~\ref{fig:laws}.
For example, consider the proof of
the @fmap@ distribution law for the lists,
also known as ``map-fusion'', which is the
basis for important optimizations in
GHC~\cite{ghc-map-fusion}.
We reflect the definition of @fmap@:
\begin{code}
  reflect map :: (a -> b) -> [a] -> [b]
  map f []     = []
  map f (x:xs) = f x : fmap f xs
\end{code}
and then specify fusion and verify it by an inductive proof:
% by induction (recursion) on the list argument:
%
\begin{mcode}
  map_fusion
    :: f:(b -> c) -> g:(a -> b) -> xs:[a]
    -> {map (f . g) xs = (map f . map g) xs}
\end{mcode}

%%\begin{mcode}
%%  map_fusion f g []
%%    =  ((map f) #.# (map g)) []
%%    =. (map f) (#map# g [])
%%    =. #map# f []
%%    =. []
%%    =. #map# (f . g) []
%%    ** QED
%%
%%  map_fusion f g (x:xs)
%%    =   #map# (f . g) (x:xs)
%%    =. (f . g) x : map (f . g) xs
%%    =. (f . g) x : (map f #.# map g) xs
%%       $\because$  map_fusion f g xs
%%    =. (f# . #g) x : map f (map g xs)
%%    =. f   (g  x): map f (map g xs)
%%    =. #map# f (g x: map g xs)
%%    =. map f   (#map# g (x:xs))
%%    =. (map f #.# map g)(x:xs)
%%    ** QED
%%\end{mcode}
%%
% \NV{Say why we need defunctionalization}
% \NV{We use app function like HALO (link to the theory)}

% \mypara{Applicative}

\mypara{Monad Laws}
The monad laws, which relate the
properties of the two operators
$\ebind$ and $\ereturn$ (Figure~\ref{fig:laws}),
refer to $\lambda$-functions,
thus their proof exercises
our support for defunctionalization
and $\eta$- and $\beta$-equivalence.
% and the extensionality axioms to prove.
%
For example, consider the proof of the
associativity law for the list monad.
First, we reflect the bind operator:
\begin{code}
  reflect (>>=) :: [a] -> (a -> [b]) -> [b]
  (x:xs) >>= f = f x ++ (xs >>= f)
  []     >>= f = []
\end{code}
Next, we define an abbreviation for the associativity property:
\begin{code}
type AssocLaw m f g =
  {m >>= f >>= g = m >>= (\x -> f x >>= g)}
\end{code}
Finally, we can prove that the list-bind is associative:
\begin{mcode}
assoc :: m:[a] -> f:(a ->[b]) -> g:(b ->[c])
      -> AssocLaw m f g
assoc [] f g
  =  [] #>>=# f >>= g
  =. [] #>>=# g
  =. []
  =. [] #>>=# (\x -> f x >>= g) ** QED
assoc (x:xs) f g
  =  (x:xs) #>>=# f  >>= g
  =. (f x ++ xs >>= f) >>= g
  =. (f x >>= g) ++ (xs >>= f >>= g)
     $\because$ bind_append (f x) (xs >>= f) g
  =. (f x >>= g) ++ (xs >>= \y -> f y >>= g)
     $\because$ assoc xs f g
  =. (\y -> f y >>= g) x ++
     (xs >>= \y -> f y >>= g)
     $\because$ $\beta$eq f g x
  =. (x:xs) #>>=# (\y -> f y >>= g) ** QED
\end{mcode}
Where the bind-append fusion lemma states that:
\begin{code}
bind_append ::
  xs:[a] -> ys:[a] -> f:(a -> [b]) ->
  {(xs++ys) >>= f = (xs >>= f)++(ys >>= f)}
\end{code}
Notice that the last step requires
$\beta$-equivalence on anonymous
functions, which we get by explicitly
inserting the redex in the logic,
via the following lemma with @trivial@ proof
\begin{mcode}
  $\beta$eq :: f:_ -> g:_ -> x:_ ->
     {bind (f x) g = (\y -> bind (f y) g) x}
  $\beta$eq _ _ _ = trivial
\end{mcode}

\subsection{Functional Correctness} \label{subsec:programs}

Finally, we proved correctness of two programs
from the literature: a SAT solver and a Unification
algorithm.

\mypara{SAT Solver}
We implemented and verified the simple
SAT solver used to illustrate and evaluate
the features of the dependently typed language
Zombie~\citep{Zombie}.
The solver takes as input a formula @f@
and returns an assignment that
\emph{satisfies} @f@ if one exists.
\begin{code}
solve :: f:Formula -> Maybe {a:Asgn|sat a f}
solve f = find (`sat` f) (assignments f)
\end{code}
Function @assignments f@ returns all possible
assignments of the formula @f@ and @sat a f@
returns @True@ iff the assignment @a@ satisfies
the formula @f@:
\begin{code}
  reflect sat :: Asgn -> Formula -> Bool
  assignments :: Formula -> [Asgn]
\end{code}
Verification of @solve@ follows simply by
reflecting @sat@ into the refinement logic,
and using (bounded) refinements to show
that @find@ only returns values on which
its input predicate yields @True@~\cite{Vazou15}.
\begin{code}
  find :: p:(a -> Bool) -> [a]
       -> Maybe {v:a | p v}
\end{code}

\mypara{Unification}
As another example, we verified the
unification of first order terms, as
presented in~\citep{Sjoberg2015}.
First, we define a predicate alias for
when two terms @s@ and @t@ are equal
under a substitution @su@:
\begin{code}
  eq_sub su s t = apply su s == apply su t
\end{code}
Now, we can define a Haskell function
@unify s t@ that can diverge, or return
@Nothing@, or return a substitution @su@
that makes the terms equal:
\begin{code}
  unify :: s:Term -> t:Term
        -> Maybe {su| eq_sub su s t}
\end{code}
For the specification and verification
we only needed to reflect @apply@ and
not @unify@; thus we only had to verify
that the former terminates, and not the latter.
%
% not that @unify@ terminates, which is a
% complicate proof.

%%% HERE
%
As before, we prove correctness by invoking
separate helper lemmas.
For example to prove the post-condition
when unifying a variable @TVar i@ with
a term @t@ in which @i@ \emph{does not}
appear, we apply a lemma @not_in@:
\begin{mcode}
  unify (TVar i) t2
    | not (i Set_mem freeVars t2)
    = Just (const [(i, t2)] $\because$ not_in i t2)

\end{mcode}
\ie if @i@ is not free in @t@,
the singleton substitution yields @t@:
\begin{code}
  not_in :: i:Int
         -> t:{Term | not (i Set_mem freeVars t)}
         -> {eq_sub [(i, t)] (TVar i) t}
\end{code}
%%
%% \NV{Emphasize how real world - diverging code co-exists with refinement reflection}

\section{Related Work}\label{sec:related}

\NV{CHECL Relational-F*, Barthe et al, from POPL 2014, and EasyCrypt}

% We compare refinement reflection to the most closely related
% lines of work in the vast literature on program verification.

\mypara{SMT-Based Verification}
SMT-solvers have been extensively used to automate
program verification via Floyd-Hoare logics~\cite{Nelson81}.
Our work is inspired by Dafny's Verified
Calculations~\citep{LeinoPolikarpova16},
a framework for proving theorems in
Dafny~\citep{dafny}, but differs in
(1)~our use of reflection instead of axiomatization, and
(2)~our use of refinements to compose proofs.
Dafny, and the related \fstar~\citep{fstar}
which like \toolname, uses types to compose
proofs, offer more automation by translating
recursive functions to SMT axioms.
However, unlike reflectionm this axiomatic
approach renders typechecking / verification
undecidable (in theory) and leads to
unpredictability and divergence
(in practice)~\citep{Leino16}.

%% In a work more closely related to
%% ours, \fstar uses refinement types
%% for program verification supporting
%% expressiveness of fully dependent types.
%% %
%% As in Dafny, \fstar directly translates
%% recursive functions to axioms in the logic
%% thus suffers from the ``butterfly effect''
%% and allows the user to explicitly write SMT tactics to control it.

%% Leino \etal~\citep{Leino16}
%% name this problem as the ``butterfly
%% effect'', in which minor modifications
%% to the program source cause significant
%% instabilities in verification and propose
%% trigger selection strategies to address it.
%% %
%% We avoid the ``butterfly effect'' by not
%% directly axiomatizing functions into logic.
%% Instead the information about the function's
%% body is exactly captured in function's result
%% type and user needs to explicitly invoke the function to push
%% the function's definition information into the logic.

\mypara{Dependent types}
Our work is inspired by dependent type
systems like Coq~\citep{coq-book} and
Agda~\citep{agda}.
Reflection shows how deep specification
and verification in the style of Coq and Agda
can be \emph{retrofitted} into existing languages
via refinement typing.
Furthermore, we can use SMT to significantly
automate reasoning over important theories like
arithmetic, equality and functions.
It would be interesting to investigate how
the tactics and sophisticated proof search
of Coq \etc can be adapted to the refinement setting.

% which allow for arbitrary expressiveness of the type system
% in the cost of automatic verification.
%
%% The syntax of \libname's operators is inspired by
%% Equational Reasoning in Agda~\citep{agda}.
%% Here we extended these equational operators
%% to support linear arithmetic and, for example, prove
%% properties of Ackermann function.
%% %
%% Unlike Adga, proof term are explicit in \libname,
%% we do not use heuristics to infer proofs.

\mypara{Dependent Types for Non-Terminating Programs}
Zombie~\citep{Zombie, Sjoberg2015} integrates
dependent types in non terminating programs
and supports automatic reasoning for equality.
Vazou \etal have previously~\citep{Vazou14} shown
how Liquid Types can be used to check
non-terminating programs.
Reflection makes \toolname at least as
expressive as Zombie, \emph{without}
having to axiomatize the theory of
equality within the type system.
Consequently, in contrast to Zombie,
SMT based reflection lets \toolname
verify higher-order specifications
like @foldr_fusion@.

% which lets us use SMT automation
% to verify deep specifications of
% non-trivial programs.
%
% Our current extension is orthogonal to the
% previous work: our system remains sound as
% long as logical terms provably terminate.
%
% We get automation from SMT solvers for not only
% the theory of equality, but also linear arithmetic.
%
% \NV{Zombie with rewritting does not allow HIGHER ORDER reasoning}

\mypara{Dependent Types in Haskell}
Integration of dependent types into Haskell
has been a long standing goal that dates back
to Cayenne~\citep{cayenne}, a Haskell-like,
fully dependent type language with undecidable
type checking.
In a recent line of work~\citep{EisenbergS14}
Eisenberg \etal aim to allow fully dependent
programming within Haskell, by making
``type-level programming ... at least as
  expressive as term-level programming''.
Our approach differs in two significant ways.
First, reflection allows SMT-aided verification
which drastically simplifies proofs over key theories
like linear arithmetic and equality.
Second, refinements are completely erased at run-time.
That is, while both systems automatically lift Haskell
code to either uninterpreted logical functions
or type families, with refinements, the logical
functions are not accessible at run-time, and
promotion cannot affect the semantics of
the program.
As an advantage (resp. disadvantage) our
proofs cannot degrade (resp. optimize)
the performance of programs.

\mypara{Proving Equational Properties}
% of Haskell Programs}
%
Several authors have proposed tools for proving
(equational) properties of (functional) programs.
Systems~\citep{sousa16} and \citep{KobayashiRelational15}
extend classical safety verification algorithms,
respectively based on Floyd-Hoare logic and Refinement Types,
to the setting of relational or $k$-safety properties
that are assertions over $k$-traces of a program.
Thus, these methods can automatically prove that
certain functions are associative, commutative \etc.
but are restricted to first-order properties and
are not programmer-extensible.
Zeno~\citep{ZENO} generates proofs by term
rewriting and Halo~\citep{HALO} uses an axiomatic
encoding to verify contracts.
Both the above are automatic, but unpredictable and not
programmer-extensible, hence, have been limited to
far simpler properties than the ones checked here.
Hermit~\citep{Farmer15} proves equalities by rewriting
GHC core guided by user specified scripts.
In contrast, our proofs are simply Haskell programs,
we can use SMT solvers to automate reasoning, and,
most importantly, we can connect the validity of
proofs with the semantics of the programs.

% \RJ{say: hermit does typeclass laws}
%
% \NV{TO ADD Naoki and class laws for TFP}
%
%% Compared to these systems, our proofs are
%% expressed as Haskell programs and do not
%% require the user to learn a different
%% tactic languages.
%% %
%% Moreover, our system is more general
%% as it allows for both equational
%% and linear arithmetic proofs.
%% %
%% On the other hand, \libname requires
%% explicit proofs and does not currently
%% support any automatic heuristics.

\section{Conclusions \& Future Directions}

We have shown how refinement reflection -- namely
reflecting the definitions of functions in their
output refinements -- can be used to convert a
language into a proof assistant, while ensuring
(refinement) type checking stays decidable and
predictable via careful design of the logic and
proof combinators.

Our evaluation shows that refinement reflection
lets us prove deep specifications of a variety
of implementations, and identifies
important avenues for research.
First, while proofs are \emph{possible}, they can
sometimes be \emph{cumbersome}. For example, in
the proof of associativity of the monadic bind
operator for the @Reader@ monad three of eight
(extensional) equalities required explanations,
some nested under multiple $\lambda$-abstractions.
Thus, it would be valuable to use recent
advances in refinement-based synthesis~\cite{polikarpova16}
to automate proof construction.
Second, while our approach to $\alpha$- and
$\beta$-equivalence is sound, we do not know
if it is \emph{complete}. We conjecture it is,
due to the fact that our refinement terms
are from the simply typed lambda calculus (STLC).
Thus, it would be interesting to use the
normalization of STLC to develop a sound
and complete SMT axiomatization, thereby
automating proofs predictably.

%% We extended Liquid Types to allow reasoning about program functions.
%% %
%% We represent program functions into logic as uninterpreted functions
%% and
%% capture the behavior of functions in the result of the function's
%% types.
%% %
%% We preserved decidable type checking by requiring the user,
%% to explicitly invoke functions,
%% instead of the solver to instantiate functional axioms.
%%
%% In the future we plan to embed existing heuristics and tactics
%% from dependent type languages, like Coq and Adga,
%% to automate common proof procedures.

%% \subsection*{Acknowledgements}

{
\bibliographystyle{plain}
\bibliography{sw}

\begin{thebibliography}{10}

\bibitem{Amin2014ComputingWA}
Nada Amin, K.~Rustan~M. Leino, and Tiark Rompf.
\newblock Computing with an {SMT} {S}olver.
\newblock In {\em TAP}, 2014.

\bibitem{KobayashiRelational15}
Kazuyuki Asada, Ryosuke Sato, and Naoki Kobayashi.
\newblock Verifying relational properties of functional programs by first-order
  refinement.
\newblock In {\em PEPM}, 2015.

\bibitem{cayenne}
L.~Augustsson.
\newblock Cayenne - a language with dependent types.
\newblock In {\em ICFP}, 1998.

\bibitem{SMTLIB2}
C.~Barrett, A.~Stump, and C.~Tinelli.
\newblock The {SMT-LIB} {S}tandard: {V}ersion 2.0.
\newblock 2010.

\bibitem{GordonRefinement09}
J.~Bengtson, K.~Bhargavan, C.~Fournet, A.D. Gordon, and S.~Maffeis.
\newblock Refinement types for secure implementations.
\newblock In {\em CSF}, 2008.

\bibitem{coq-book}
Y.~Bertot and P.~Cast\'eran.
\newblock {\em Coq'Art: The Calculus of Inductive Constructions}.
\newblock Springer Verlag, 2004.

\bibitem{Zombie}
C.~Casinghino, V.~Sj{\"o}berg, and S.~Weirich.
\newblock Combining proofs and programs in a dependently typed language.
\newblock In {\em POPL}, 2014.

\bibitem{deputy}
Jeremy Condit, Matthew Harren, Zachary~R. Anderson, David Gay, and George~C.
  Necula.
\newblock Dependent types for low-level programming.
\newblock In {\em ESOP}, 2007.

\bibitem{ConstableS87}
R.~L. Constable and S.~F. Smith.
\newblock Partial objects in constructive type theory.
\newblock In {\em LICS}, 1987.

\bibitem{EisenbergS14}
Richard~A. Eisenberg and Jan Stolarek.
\newblock Promoting functions to type families in {H}askell.
\newblock In {\em Haskell}, 2014.

\bibitem{Farmer15}
Andrew Farmer, Neil Sculthorpe, and Andy Gill.
\newblock Reasoning with the {HERMIT}: {T}ool support for equational reasoning
  on {GHC} {C}ore programs.
\newblock Haskell, 2015.

\bibitem{HuttonBook}
Graham Hutton.
\newblock {\em Programming in Haskell}.
\newblock Cambridge University Press, 2007.

\bibitem{RefinedRacket}
Andrew~M. Kent, David Kempe, and Sam Tobin{-}Hochstadt.
\newblock Occurrence typing modulo theories.
\newblock In {\em PLDI}, 2016.

\bibitem{Knowles10}
K.W. Knowles and C.~Flanagan.
\newblock Hybrid type checking.
\newblock {\em ACM TOPLAS}, 2010.

\bibitem{dafny}
K.~Rustan~M. Leino.
\newblock Dafny: An automatic program verifier for functional correctness.
\newblock LPAR, 2010.

\bibitem{Leino16}
K.~Rustan~M. Leino and Cl\'{e}ment Pit{-}Claudel.
\newblock Trigger selection strategies to stabilize program verifiers.
\newblock In {\em CAV}, 2016.

\bibitem{LeinoPolikarpova16}
K.~Rustan~M. Leino and Nadia Polikarpova.
\newblock Verified calculations.
\newblock In {\em VSTTE}, 2016.

\bibitem{agdaequational}
Shin-cheng Mu, Hsiang-shang Ko, and Patrik Jansson.
\newblock {A}lgebra of {P}rogramming in {A}gda: {D}ependent {T}ypes for
  {R}elational {P}rogram {D}erivation.
\newblock {\em J. Funct. Program.}, 2009.

\bibitem{Nelson81}
G.~Nelson.
\newblock Techniques for program verification.
\newblock Technical Report CSL81-10, Xerox Palo Alto Research Center, 1981.

\bibitem{agda}
U.~Norell.
\newblock {\em Towards a practical programming language based on dependent type
  theory}.
\newblock PhD thesis, Chalmers, 2007.

\bibitem{Ou2004}
X.~Ou, G.~Tan, Y.~Mandelbaum, and D.~Walker.
\newblock Dynamic typing with dependent types.
\newblock In {\em IFIP TCS}, 2004.

\bibitem{polikarpova16}
Nadia Polikarpova, Ivan Kuraj, and Armando Solar{-}Lezama.
\newblock Program synthesis from polymorphic refinement types.
\newblock In {\em PLDI}, 2016.

\bibitem{Reynolds72}
John~C. Reynolds.
\newblock Definitional interpreters for higher-order programming languages.
\newblock In {\em 25th ACM National Conference}, 1972.

\bibitem{LiquidPLDI08}
P.~Rondon, M.~Kawaguchi, and R.~Jhala.
\newblock Liquid types.
\newblock In {\em PLDI}, 2008.

\bibitem{LiquidPOPL10}
P.~Rondon, M.~Kawaguchi, and R.~Jhala.
\newblock Low-level liquid types.
\newblock In {\em POPL}, 2010.

\bibitem{Rushby98}
J.~Rushby, S.~Owre, and N.~Shankar.
\newblock Subtypes for specifications: Predicate subtyping in pvs.
\newblock {\em IEEE TSE}, 1998.

\bibitem{Sjoberg2015}
Vilhelm Sj\"{o}berg and Stephanie Weirich.
\newblock Programming up to congruence.
\newblock {\em POPL}, 2015.

\bibitem{ZENO}
W.~Sonnex, S.~Drossopoulou, and S.~Eisenbach.
\newblock Zeno: An automated prover for properties of recursive data
  structures.
\newblock In {\em TACAS}, 2012.

\bibitem{sousa16}
Marcelo Sousa and Isil Dillig.
\newblock Cartesian hoare logic for verifying k-safety properties.
\newblock In {\em PLDI}, 2016.

\bibitem{fstar}
Nikhil Swamy, C\u{a}t\u{a}lin Hri\c{t}cu, Chantal Keller, Aseem Rastogi,
  Antoine Delignat-Lavaud, Simon Forest, Karthikeyan Bhargavan, C\'{e}dric
  Fournet, Pierre-Yves Strub, Markulf Kohlweiss, Jean-Karim Zinzindohoue, and
  Santiago Zanella-B\'eguelin.
\newblock Dependent types and multi-monadic effects in {F*}.
\newblock In {\em POPL}, 2016.

\bibitem{ackermann}
George Tourlakis.
\newblock {A}ckermann’s {F}unction.
\newblock \url {http://www.cs.yorku.ca/~gt/papers/Ackermann-function.pdf},
  2008.

\bibitem{Vazou14}
N.~Vazou, E.~L. Seidel, R.~Jhala, D.~Vytiniotis, and S.~Peyton-Jones.
\newblock {R}efinement {T}ypes for {H}askell.
\newblock In {\em ICFP}, 2014.

\bibitem{Vazou15}
Niki Vazou, Alexander Bakst, and Ranjit Jhala.
\newblock Bounded refinement types.
\newblock In {\em ICFP}, 2015.

\bibitem{Vekris16}
Panagiotis Vekris, Benjamin Cosman, and Ranjit Jhala.
\newblock Refinement types for typescript.
\newblock In {\em PLDI}, 2016.

\bibitem{HALO}
D.~Vytiniotis, S.L. Peyton-Jones, K.~Claessen, and D.~Ros{\'e}n.
\newblock Halo: haskell to logic through denotational semantics.
\newblock In {\em POPL}, 2013.

\bibitem{ghc-map-fusion}
GHC Wiki.
\newblock {GHC} optimisations.
\newblock \url{https://wiki.haskell.org/GHC_optimisations}.

\bibitem{pfenningxi98}
H.~Xi and F.~Pfenning.
\newblock Eliminating array bound checking through dependent types.
\newblock In {\em PLDI}, 1998.

\end{thebibliography}
}

\ifthenelse{\equal{\isTechReport}{true}}
{
\appendix
   \section{Implementation: \href{https://github.com/ucsd-progsys/liquidhaskell/tree/popl17}{Liquid Haskell}}

Refinement Reflection is fully implemented in 
Liquid Haskell and will be included in the next release. 
The implementation can be found in the
\href{https://github.com/ucsd-progsys/liquidhaskell/tree/popl17}{github repository}
and all the benchmarks of Sections~\ref{sec:overview} and~\ref{sec:evaluation}
are included in our 
\href{https://github.com/ucsd-progsys/liquidhaskell/tree/popl17/benchmarks/popl17/pos}{tests}. 
Next, we describe the file 
\href{https://github.com/ucsd-progsys/liquidhaskell/blob/popl17/benchmarks/popl17/pos/Proves.hs}{Proves.hs}, 
the library of proof combinators used by our benchmarks
and discuss known limitations of our implementation.

\renewcommand\libname{\ensuremath{\texttt{Proves}}\xspace}

\subsection{\libname: The Proof Combinators Library}
\label{subsec:library}

In this section we present \libname,
a Haskell library used to structure proof terms.
\libname is inspired by Equational Reasoning Data Types
in Adga~\citep{agdaequational}, providing operators to
construct proofs for equality and linear arithmetic in Haskell.
The constructed proofs are checked by an SMT-solver via Liquid Types.

\spara{Proof terms} are defined in \libname as a type alias for unit,
a data type that curries no run-time information
\begin{code}
  type Proof = ()
\end{code}
Proof types are refined to express theorems about program functions.
For example, the following @Proof@ type expresses that
@fib 2 == 1@
\begin{code}
  fib2 :: () -> {v:Proof | fib 2 == 1}
\end{code}
We simplify the above type by omitting the irrelevant
basic type @Proof@ and variable @v@
\begin{code}
  fib2 :: () -> { fib 2 == 1 }
\end{code}

\libname provides primitives to construct proof terms
by casting expressions to proofs.
To resemble mathematical proofs,
we make this casting post-fix.
We write @p *** QED@ to cast @p@ to a proof term,
by defining two operators @QED@ and @***@ as
\begin{code}
  data QED = QED

  (***) :: a -> QED -> Proof
  _ *** _  = ()
\end{code}

\spara{Proof construction.}
To construct proof terms, \libname
provides a proof constructor @op.@
for logical operators of the theory of
linear arithmetic and equality:
$\{=, \not =, \leq, <, \geq, > \} \in \odot$.
@op. x y@ ensures that $x \odot y$ holds, and returns @x@
\begin{code}
  op.:: x:a -> y:{a| x op y} -> {v:a| v==x}
  op. x _ = x

  -- for example
  ==.:: x:a -> y:{a| x==y} -> {v:a| v==x}
\end{code}
For instance, using @==.@
we construct a proof, in terms of Haskell code,
that @fib 2 == 1@:
\begin{code}
  fib2 _
    =   fib 2
    ==. fib 1 + fib 0
    ==. 1
    *** QED
\end{code} %$

\spara{Reusing proofs: Proofs as optional arguments.}
Often, proofs require reusing existing proof terms.
For example, to prove @fib 3 == 2@ we can reuse the above
@fib2@ proof.
We extend the proof combinators, to receive
an \textit{optional} third argument of @Proof@ type.
\begin{code}
  op.:: x:a -> y:a -> {x op y} -> {v:a|v==x}
  op. x _ _ = x
\end{code}
@op. x y p@ returns @x@ while the third argument @p@
explicitly proves $x \odot y$.

\spara{Optional Arguments.}
The proof term argument is optional.
To implement optional arguments in Haskell we use the standard technique
where for each operator @op!@ we define a type class @Optop@
that takes as input two expressions @a@ and returns a result @r@,
which will be instantiated with either the result value @r:=a@
or a function form a proof to the result @r:=Proof ->  a@.
\begin{code}
  class Optop a r where
    (op.) :: a -> a -> r
\end{code}
When no explicit proof argument is required,
the result type is just an @y:a@ that curries the proof @x op y@
\begin{code}
  instance Optop a a where
  (op.) :: x:a->y:{a| x op y}->{v:a | v==x }
  (op.) x _ = x
\end{code}
Note that Haskell's type inference~\citep{Sulzmann06}
requires both type class parameters @a@ and @r@ to be constrainted at class instance
matching time.
In most our examples, the result type parameter @r@ is not constrained
at instance matching time, thus
due to the Open World Assumption
the matching instance could not be determined.
To address the above, we used another common Haskell trick,
of generalizing the instance to type arguments @a@ and @b@ and then
constraint @a@ and @b@ to be equal @a~b@.
This generalization allows the instance to always match and
imposed the equality constraint after matching.
\begin{code}
  instance (a~b)=>Optop a b where
  (op.) :: x:a->y:{x op y}->{v:b | v==x }
  (op.) x _ = x
\end{code}

To explicitly provide a proof argument,
the result type @r@ is instantiated to @r:= Proof -> a@.
For the same instance matching restrictions as above,
the type is further generalized to return some @b@
that is constraint to be equal to @a@.
\begin{code}
  instance (a~b)=>Optop a (Proof->b) where
  (op.) :: x:a->y:a->{x op y}->{v:b | v==x }
  (op.) x _ _ = x
\end{code}
As a concrete example, we define the equality operator @==.@
via the type class @OptEq@ as
\begin{code}
  class OptEq a r where
   (==.):: a -> a -> r

  instance (a~b)=>OptEq a b where
   (==.)::x:a->y:{a|x==y}->{v:b|v==x}
   (==.) x _ = x

  instance (a~b)=>OptEq a (Proof->b) where
   (==.)::x:a->y:a->{x==y}->{v:b|v==x}
   (==.) x _ _ = x
\end{code}

\spara{Explanation Operator.}
The ``explanation operator'' @(?)@, or @($\because$)@, 
is used to better structure the proofs.
@(?)@ is an infix operator with same fixity as @(op.)@
that allows for the equivalence
@ x op. y ? p == (op.) x y p@
\begin{code}
  (?) :: (Proof -> a) -> Proof -> a
  f ? y = f y
\end{code}

\spara{Putting it all together}
Using the above operators,
we prove that @fib 3 == 2@,
reusing the previous proof of @fib 2 == 1@,
in a Haskell term that resembles mathematical proofs
\begin{code}
  fib3 :: () ->  {fib 3 == 2}
  fib3 _
    =   fib 3
    ==. fib 2 + fib 1
    ==. 2             ? fib2 ()
    *** QED
\end{code}

\spara{Unverified Operators}
All operators in \libname, but two are implemented in Haskell
with implementations verified by Liquid Haskell. 
The ''unsound`` operators are the assume 
(1). @(==?)@ that eases proof construction by assuming equalities, to be proven later
and (2). @(=*)@ extentional proof equality. 

\spara{Assume Operator} @(==?)@ eases proof construction by 
assuming equalities while the proof is in process. 
It is not implemented in that its body is @undefined@. 
Thus, if we run proof terms including assume operator, the proof will merely crash
(instead of returning @()@). 
Proofs including the assume operator are not considered complete, 
as via assume operator any statement can be proven, 

\spara{Function Extensional Equality}
Unlike the assume operator that is undefined and 
included in unfinished thus unsound proofs, 
the functions extensionality is included in valid proofs
that assume function extensionality, an axioms that is assumed, 
as it cannot be proven by our logic. 

To allow function equality via extensionality, 
we provide the user with a function comparison operator
that for each function @f@ and @g@ it transforms a proof 
that for every argument @x@, @f x = g x@ to a proof 
on function equality @f = g@. 
\begin{code}
(=*) :: Arg a => f:(a -> b) -> g:(a -> b)
     -> p:(x:a -> {f x = g x})
     -> {f = g}
\end{code}
The function @(=*)@ is not implemented in the library: 
it returns () and its type is assumed. 
But soundness of its usage requires the argument type variable @a@
to be constrained by a type class constraint @Arg a@, 
for both operational and type theoretic reasons.

From \textit{operational} point of view,
an implementation of @(=*)@ would require checking 
equality of @f x = g x@ \textit{forall} arguments @x@ of type @a@. 
This equality would hold due to the proof argument @p@. 
The only missing point is a way to enumerate all the argument @a@, 
but this could be provided by a method of the type clas @Arg a@. 
Yet, we have not implement @(=*)@ because we do not know how to 
provide such an implementation that can provably satisfy @(=*)@'s type.

From \textit{type theoretic} point of view, 
the type variable argument @a@
appears only on negative positions. 
Liquid type inference is smart enough to infer that 
since @a@ appears only negative @(=*)@ cannot use any @a@
and thus will not call any of its argument arguments @f@, @g@, nor the @p@. 
Thus, at each call site of @(=*)@ the type variable `a` is instantiated
with the refinement type @{v:a | false}@ indicating dead-code 
(since @a@s will not be used by the callee.)
Refining the argument @x:a@ with false at each call-site though
leads to unsoundness, as each proof argument @p@ is a valid proof under 
the false assumption. 
What Liquid inference cannot predict is our intention to call 
@f@, @g@ and @p@ at \textit{every possible argument}. 
This information is capture by the type class constraint @Arg a@
that (as discussed before~\citep{Vazou13}) states that methods of 
the type class 
@Arg a@ may create values of type @a@, thus, 
due to lack of information on the values that are created by the
methods of @Arg a@, @a@ can only be refined with @True@.

With extensional equality, we can prove 
that @\x -> x@ is equal to @\x -> id x@, 
by providing an explicit explanation that 
if we call both these functions with the same 
argument @x@, they return the same result, for each @x@.
\begin{mcode}
 safe :: Arg a => a 
       -> {(\x -> id x) = (\x -> x)}
 safe _  = (\x -> x) 
         =*(\x -> id x) $\because$ (exp ())
 
 exp :: Arg a => a -> x:a 
     -> {(\x -> id x) x = (\x -> x) x}
 exp _ x =  id x 
         ==. x
         *** QED
\end{mcode} 
Note that the result of @exp@ 
is an equality of the redexes 
@(\x -> id x) x@ and @((\x -> x) x@. 
Extentional function equality requires 
as argument an equality on such redexes. 
Via $\beta$ equality instantiations, 
both such redexes will automatically reduce, 
requiring @exp@ to prove @id x = x@, 
with is direct.  

Admittedly, proving function equality via extensionality
is requires a cumbersome indirect proof. 
For each function equality in the main proof
one needs to define an explanation function 
that proves the equality for every argument.

\subsection{Engineering Limitations}
The theory of refinement reflection is fully implemented in Liquid Haskell. 
Yet, to make this extension \textit{usable} in \textit{real world applications}
there are four known engineering limitations that need to be addressed. 
All these limitations seem straightforward to address 
and we plan to fix them soon.

\spara{The language of refinements}
is typed lambda calculus. That is the types of the lambda arguments are 
explicitly specified 
instead of being inferred. 
As another minor limitation, the refinement language parser 
requires the argument to be enclosed in parenthesis
in applications where the function is not a variable. 
Thus the Haskell expression 
@(\x -> x) e@ should be written as @(\x:a -> x) (e)@
in the refinement logic, 

\spara{Class instances methods} can not be reflected. 
Instead, the methods we want to use in the theorems/propositions
should be defined as Haskell functions. 
This restriction has two major implications. 
Firstly, we can not verify correctness of library provided 
instances but we need to redifine them ourselves. 
Secondly, we cannot really verify class instances with class preconditions. 
For example, during verification of monoid associativity of the Maybe instance
\begin{code}
  instance (Monoid a) => Monoid (Maybe a)
\end{code}
there is this @Monoid a@ class constraint assumption we needed to raise to 
proceed verification.

\spara{Only user defined data types}
can currently used in verification. 
The reason for this limitation is that 
reflection of case expressions 
requires checker and projector measures for each 
data type used in reflected functions. 
Thus, not only should these data types be defined in 
the verified module, but also should be 
be injected in the logic by providing a refined version of 
the definition that can (or may not) be trivially refined. 

For example, to reflect a function that uses 
@Peano@ numbers, the Haskell \textit{and} the refined @Peano@ 
definitions should be provided
\begin{code}
data Peano = Z | S Peano

{-@ data Peano [toInt] 
     = Z 
     | S {prev :: Peano} 
  @-}
\end{code}
Note that the termination function @toInt@ 
that maps @Peano@ numbers to natural numbers 
is also crucial for soundness of reflection. 

\spara{There is no module support.}
All reflected definitions, 
including, measures (automatically generated checkers and selector, 
but also the classic lifted Haskell functions to measures) 
and the reflected types of the reflected functions, 
are not exposed outside of the module they are defined. 
Thus all definitions and propositions should exist in the same module. 

   \section{Soundness of Refinement Reflection}

We prove Theorem~\ref{thm:safety}
of \S~\ref{sec:types-reflection}
by reduction to Soundness of \undeclang~\citep{Vazou14}. 

\begin{theorem}{[Denotations]}~\label{tech:thm:denotations}
If $\hastype{\env}{\prog}{\typ}$ then
$\forall \sto\in \interp{\env}. \applysub{\sto}{\prog} \in \interp{\applysub{\sto}{\typ}}$.
\end{theorem}
\begin{proof}
We use the proof from~\citep{Vazou14-tech} and specifically Lemma 4
that is identical to the statement we need to prove. 
Since the proof proceeds by induction in the type derivation, 
we need to ensure that all the modified rules satisfy the statement. 
\begin{itemize}
\item\rtexact
 Assume 
 	\hastype{\env}{e}{\tref{v}{\btyp}{\reft_r\land v = e}}.
 By inversion
	\hastype{\env}{e}{\tref{v}{\btyp}{\reft_r}}(1). 
 By (1) and IH we get 
 $\forall \sto\in \interp{\env}. 
   \applysub{\sto}{e} \in \interp{\applysub{\sto}{\tref{v}{\btyp}{\reft_r}}}$.
 We fix a $\sto\in \interp{\env}$
 We get that if \evalsto{\applysub{\sto}{e}}{w}, 
 then $\evalsto{\applysub{\sto}{\reft_r}\subst{v}{w}}{\etrue}$.  
 By the Definition of $=$ we get that 
 $\evalsto{w = w}{\etrue}$. 
 Since $\evalsto{\applysub{\sto}{(v = e)}\subst{v}{w}}{w = w}$, 
 then $\evalsto{\applysub{\sto}{(\reft_r\land v = e)}\subst{v}{w}}{\etrue}$.  
 Thus
   $\applysub{\sto}{e} \in \interp{\applysub{\sto}{\tref{v}{\btyp}{\reft_r\land v = e}}}$
  and since this holds for any fixed $\sto$,  
 $\forall \sto\in \interp{\env}. 
   \applysub{\sto}{e} \in \interp{\applysub{\sto}{\tref{v}{\btyp}{\reft_r\land v = e}}}$.
\item\rtlet
  Assume 
	\hastype{\env}{\eletb{x}{\gtyp_x}{e_x}{\prog}}{\typ}.  
  By inversion
	\hastype{\env, \tbind{x}{\gtyp_x}}{e_x}{\gtyp_x} (1), 
	\hastype{\env, \tbind{x}{\gtyp_x}}{\prog}{\gtyp} (2), and
    \iswellformed{\env}{\typ} (3). 
 By IH 
	$\forall \sto\in \interp{\env, \tbind{x}{\gtyp_x}}. 
	\applysub{\sto}{e_x} \in \interp{\applysub{\sto}{\gtyp_x}}$ (1')
	$\forall \sto\in \interp{\env, \tbind{x}{\gtyp_x}}. 
	\applysub{\sto}{\prog} \in \interp{\applysub{\sto}{\gtyp}}$ (2'). 
 By (1') and by the type of $\efix{}$ 
	$\forall \sto\in \interp{\env, \tbind{x}{\gtyp_x}}. 
	\applysub{\sto}{\efix{x}\ e_x} \in \interp{\applysub{\sto}{\gtyp_x}}$. 
 By which,  (2') and (3)
	$\forall \sto\in \interp{\env}. 
	\applysub{\sto}{\SUBST{\prog}{x}{\efix{x}\ {e_x}}} \in \interp{\applysub{\sto}{\gtyp}}$.  
%% \NV{CHECK}
\item\rtreflect
  Assume 
  \hastype{\env}{\erefb{f}{\gtyp_f}{e}{\prog}}
			    {\typ}. 
  By inversion, 
    \hastype{\env}{\eletb{f}{\exacttype{\gtyp_f}{e}}{e}{\prog}}
			     {\typ}. 
  By IH, 
 	$\forall \sto\in \interp{\env}. 
	\applysub{\sto}{\eletb{f}{\exacttype{\gtyp_f}{e}}{e}{\prog}} \in \interp{\applysub{\sto}{\gtyp}}$.  
  Since denotations are closed under evaluation, 
	$\forall \sto\in \interp{\env}. 
	\applysub{\sto}{\erefb{f}{\exacttype{\gtyp_f}{e}}{e}{\prog}} \in \interp{\applysub{\sto}{\gtyp}}$.  

\item\rtfix
  In Theorem 8.3 from~\citep{Vazou14-tech} (and using the textbook proofs from~\citep{PLC})
  we proved that for each type $\typ$, $\efix{}_\typ \in \interp{(\typ \rightarrow \typ) \rightarrow \typ}$.
\end{itemize}
\end{proof}

\begin{theorem}{[Preservation]}
If \hastype{\emptyset}{\prog}{\typ}
       and $\evalsto{\prog}{w}$ then $\hastype{\emptyset}{w}{\typ}$.
\end{theorem}
\begin{proof}
In~\citep{Vazou14-tech} proof proceeds by iterative application 
of Type Preservation Lemma 7. 
Thus, it suffices to ensure Type Preservation in \corelan, which 
it true by the following Lemma.
\end{proof}

\begin{lemma}
If \hastype{\emptyset}{\prog}{\typ}
       and $\evals{\prog}{\prog'}$ then $\hastype{\emptyset}{\prog'}{\typ}$.
\end{lemma}
\begin{proof}
Since Type Preservation in \undeclang is proved by induction on the type derivation tree, 
we need to ensure that all the modified rules satisfy the statement. 
\begin{itemize}
\item\rtexact
 Assume 
 	\hastype{\emptyset}{\prog}{\tref{v}{\btyp}{\reft_r\land v = \prog}}.
 By inversion
	\hastype{\emptyset}{\prog}{\tref{v}{\btyp}{\reft_r}}.
 By IH we get 
	\hastype{\emptyset}{\prog'}{\tref{v}{\btyp}{\reft_r}}.
 By rule \rtexact we get 
 	\hastype{\emptyset}{\prog'}{\tref{v}{\btyp}{\reft_r\land v = \prog'}}.
 Since subtyping is closed under evaluation, we get 
 	\issubtype{\emptyset}{\tref{v}{\btyp}{\reft_r\land v = \prog'}}
 	                {\tref{v}{\btyp}{\reft_r\land v = \prog}}.
 By rule \rtsub we get 
 	\hastype{\emptyset}{\prog'}{\tref{v}{\btyp}{\reft_r\land v = \prog}}.

\item\rtlet
  Assume 
	\hastype{\emptyset}{\eletb{x}{\gtyp_x}{e_x}{\prog}}{\typ}.
 By inversion, 
   \hastype{\tbind{x}{\gtyp_x}}{e_x}{\gtyp_x}  (1), 
   \hastype{\tbind{x}{\gtyp_x}}{\prog}{\gtyp} (2), and
   \iswellformed{\env}{\typ} (3). 
 By rule \rtfix
   \hastype{\tbind{x}{\gtyp_x}}{\efix{x}\ {e_x}}{\gtyp_x}  (1').
 By (1'), (2) and Lemma 6 of~\citep{Vazou14-tech}, we get 
   \hastype{}{\SUBST{\prog}{x}{\efix{x}\ {e_x}}}{\SUBST{\gtyp}{x}{\efix{x}\ e_x}}. 
 By (3)
   $ \SUBST{\gtyp}{x}{\efix{x}\ e_x} \equiv \gtyp$.    
 Since 
   $\prog' \equiv \SUBST{\prog}{x}{\efix{x}\ {e_x}}$, 
 we have 
 \hastype{\emptyset}{\prog'}{\gtyp}. 

\item\rtreflect
  Assume 
	\hastype{\emptyset}{\erefb{x}{\gtyp_x}{e_x}{\prog}}{\typ}.
 By double inversion, with $\gtyp_x' \equiv \exacttype{\gtyp_x}{e_x} $; 
   \hastype{\tbind{x}{\gtyp_x'}}{e_x}{\gtyp_x'}  (1), 
   \hastype{\tbind{x}{\gtyp_x'}}{\prog}{\gtyp} (2), and
   \iswellformed{\env}{\typ} (3). 
 By rule \rtfix
   \hastype{\tbind{x}{\gtyp_x'}}{\efix{x}\ {e_x}}{\gtyp_x'}  (1').
 By (1'), (2) and Lemma 6 of~\citep{Vazou14-tech}, we get 
   \hastype{}{\SUBST{\prog}{x}{\efix{x}\ {e_x}}}{\SUBST{\gtyp}{x}{\efix{x}\ e_x}}. 
 By (3)
   $ \SUBST{\gtyp}{x}{\efix{x}\ e_x} \equiv \gtyp$.    
 Since 
   $\prog' \equiv \SUBST{\prog}{x}{\efix{x}\ {e_x}}$, 
 we have 
 \hastype{\emptyset}{\prog'}{\gtyp}. 

\item\rtfix
  This case cannot occur, as $\efix{}$ does not evaluate to any program. 
\end{itemize}
\end{proof}

   \section{Soundness of Algorithmic Verification}
In this section we prove soundness of Algorithmic verification, 
by proving the theorems of \S~\ref{sec:algorithmic}
by referring to the proofs in~\citep{Vazou14-tech}. 

\subsection{Transformation}

\begin{definition}[Initial Environment]\label{def:initialsmt}
 We define the initial SMT environment \smtenvinit to include
 $$
 \begin{array}{rcll}
 \smtvar{c}  &\colon &\embed{\constty{c}}
   &\forall c\in \corelan\\
 \smtlamname{\sort_x}{\sort}&\colon&\sort_x \rightarrow \sort\rightarrow\tsmtfun{\sort_x}{\sort}
   &\forall \sort_x, \sort\in \smtlan\\
 \smtappname{\sort_x}{\sort}&\colon&\tsmtfun{\sort_x}{\sort} \rightarrow \sort_x \rightarrow \sort
   &\forall \sort_x, \sort\in \smtlan\\
 \smtvar{\dc}&\colon&\embed{\constty{\dc}}
   &\forall\dc\in\corelan\\
 \checkdc{\dc}&\colon&\embed{T \rightarrow \tbool}
   &\forall \dc\in \corelan\ \text{of data type}\ T \\
 \selector{\dc}{i}&\colon&\embed{T \rightarrow \typ_i}
   &\forall \dc\in \corelan\ \text{of data type}\ T \\
   &&&\text{and}\ i\text{-th argument}\ \typ_i \\
 {x^{\sort}_i} & \colon&{\sort}& \forall \sort \in \smtlan \text{and} 1\leq i\leq \maxlamarg\\
 \end{array}
 $$
Where $x^{\sort}_i$ are $\maxlamarg$ global names that only appear as lambda arguments.
\end{definition}

We modify the $\lgfun$ rule to ensure that 
logical abstraction is performed 
using the minimum globally defined lambda argument that is not already abstracted. 
We do so, using the helper function \maxlam{\sort}{\pred}:

\begin{align*}
\maxlam{\sort}{\smtlamname{\sort}{\sort'}\ {x^{\sort}_i}\ \pred} =& \mathtt{max}(i, \maxlam{\sort}{\pred})\\
\maxlam{\sort}{r\ \overline{r}} =& \mathtt{max}(\maxlam{\sort}{\pred, \overline{\pred}}) \\
\maxlam{\sort}{\pred_1 \binop \pred_2} = &  \mathtt{max}(\maxlam{\sort}{\pred_1, \pred_2})\\
\maxlam{\sort}{\unop \pred} =& \maxlam{\sort}{\pred}\\
\maxlam{\sort}{\eif{\pred}{\pred_1}{\pred_2}} =& \mathtt{max}(\maxlam{\sort}{\pred, \pred_1, \pred_2}) \\
\maxlam{\sort}{\pred} =& 0 \\
\end{align*}
$$
\inference{
    i = \maxlam{\embed{\typ_x}}{\pred} & i < \maxlamarg & y = x^{\embed{\typ_x}}_{i+1} \\ 
    \tologicshort{\env, \tbind{y}{\typ_x}}{e\subst{x}{y}}{}{\pred}{}{}{} &
  	\hastype{\env}{(\efun{x}{}{e})}{(\tfun{x}{\typ_x}{\typ})}\\
}{
	\tologicshort{\env}{\efun{x}{}{e}}{}
	        {\smtlamname{\embed{\typ_x}}{\embed{\typ}}\ {y}\ {\pred}}
	        {}{}{}
}[\lgfun]
$$

\begin{lemma}[Type Transformation]
If \tologicshort{\env}{e}{\typ}{p}{\sort}{\smtenv}{\axioms},
and \hastype{\env}{e}{\typ}, then
\smthastype{\smtenvinit, \embed{\env}}{p}{\embed{\typ}}.
\end{lemma}
\begin{proof}
We proceed by induction on the translation
\begin{itemize}
\item \lgbool : Since $\embed{\tbool} = \tbool$,
If \hastype{\env}{b}{\tbool}, then  
\smthastype{\smtenvinit, \embed{\env}}{b}{\embed{\tbool}}.

\item \lgint :  Since $\embed{\tint} = \tint$,
If \hastype{\env}{n}{\tint}, then  
\smthastype{\smtenvinit, \embed{\env}}{n}{\embed{\tint}}.

\item \lgun : 
Since $\hastype{\env}{\lnot\ e}{\typ}$, then it should be 
$\hastype{\env}{e}{\tbool}$ and $\typ \equiv \tbool$.
By IH,  
\smthastype{\smtenvinit, \embed{\env}}{\pred}{\embed{\tbool}}, 
thus 
\smthastype{\smtenvinit, \embed{\env}}{\lnot \pred}{\embed{\tbool}}. 

\item \lgbinGEN
Assume 	
 $\tologicshort{\env}{e_1\binop e_2}{\tbool}{\pred_1 \binop\pred_2}{\tbool}{\smtenv}{\andaxioms{\axioms_1}{\axioms_2}}$.  
By inversion
    \tologicshort{\env}{e_1}{\typ}{\pred_1}{\embed{\typ}}{\smtenv}{\axioms_1}, and
    \tologicshort{\env}{e_2}{\typ}{\pred_2}{\embed{\typ}}{\smtenv}{\axioms_2}.
Since 
  \hastype{\env}{e_1\binop e_2}{\typ}, then 
  \hastype{\env}{e_1}{\typ_1} and  
  \hastype{\env}{e_1}{\typ_2}.
By IH,  
\smthastype{\smtenvinit, \embed{\env}}{\pred_1}{\embed{\typ_1}} and  
\smthastype{\smtenvinit, \embed{\env}}{\pred_2}{\embed{\typ_2}}.
We split cases on $\binop$
\begin{itemize}
\item If $\binop \equiv =$, then 
  $\typ_1 = \typ_2$, thus $\embed{\typ_1} = \embed{\typ_2}$
  and $\embed{\typ} = \typ = \tbool$.

\item If $\binop \equiv <$, then 
  $\typ_1 = \typ_2 = \tint$, thus $\embed{\typ_1} = \embed{\typ_2} = \tint$
  and $\embed{\typ} = \typ = \tbool$.

\item If $\binop \equiv \land$, then 
  $\typ_1 = \typ_2 = \tbool$, thus $\embed{\typ_1} = \embed{\typ_2} = \tbool$
  and $\embed{\typ} = \typ = \tbool$.
\item If $\binop \equiv +$ or $\binop \equiv -$, then 
  $\typ_1 = \typ_2 = \tint$, thus $\embed{\typ_1} = \embed{\typ_2} = \tint$
  and $\embed{\typ} = \typ = \tint$.
\end{itemize}

\item \lgvar : 
Assume 
	\tologicshort{\env}{x}{\env(x)}{x}{\embed{\env(x)}}{\emptyset}{\emptyaxioms}
Then 
   \hastype{\env}{x}{\env (x)} and 
   \smthastype{\smtenvinit, \embed{\env}}{x}{\embed{\env} (x)}.
But by definition 
  $(\embed{\env}) (x) = \embed{\env(x)}$. 

\item \lgpop : 
Assume 
	\tologicshort{\env}{c}{\constty{\odot}}{\smtvar{c}}{\embed{\constty{\odot}}}{\emptyset}{\emptyaxioms}
Also, 
   \hastype{\env}{c}{\constty{c}} and 
   \smthastype{\smtenvinit, \embed{\env}}{\smtvar{c}}{\smtenvinit (\smtvar{c})}.
But by Definition~\ref{def:initialsmt}
  $\smtenvinit (\smtvar{c}) = \embed{\constty{c}}$. 

\item \lgdc : 
Assume 
	\tologicshort{\env}{\dc}{\constty{\dc}}{\smtvar{\dc}}{\embed{\constty{\dc}}}{\emptyset}{\emptyaxioms}
Also, 
   \hastype{\env}{\dc}{\constty{\dc}} and 
   \smthastype{\smtenvinit, \embed{\env}}{\smtvar{\dc}}{\smtenvinit (\smtvar{\dc})}.
But by Definition~\ref{def:initialsmt}
  $\smtenvinit(\smtvar{\dc}) = \embed{\constty{c}}$. 

\item \lgfun : 
Assume 
	\tologicshort{\env}{\efun{x}{}{e}}{(\tfun{x}{\typ_x}{\typ})}
	        {\smtlamname{\embed{\typ_x}}{\embed{\typ}}\ {x^{\embed{\typ_x}}_{i}}\ {\pred}}
	        {\sort'}{\smtenv, \tbind{f}{\sort'}}{\andaxioms{\{\axioms_{f_1}, \axioms_{f_2}\}}{\axioms}}. 
By inversion 
    $i \leq \maxlamarg$
	\tologicshort{\env, \tbind{x^{\embed{\typ_x}}_{i}}{\typ_x}}{e\subst{x}{x^{\embed{\typ_x}}_{i}}}{}{\pred}{}{}{}, and
  	\hastype{\env}{(\efun{x}{}{e})}{(\tfun{x}{\typ_x}{\typ})}.
By the Definition~\ref{def:initialsmt} on $\smtlamname{}{}$, $x^{\sort}_i$ and induction, we get
   \smthastype{\smtenvinit, \embed{\env}}
     {\smtlamname{\embed{\typ_x}}{\embed{\typ}}\ {x^{\embed{\typ_x}}_{i}}\ {\pred}}
     {\tsmtfun{\embed{\typ_x}}{\embed{\typ}}}.
By the definition of the type embeddings we have
$\embed{\tfunbasic{x}{\typ_x}{\typ}} \defeq \tsmtfun{\embed{\typ_x}}{\embed{\typ}}$.

\item \lgapp : 
Assume 
\tologicshort{\env}{e\ e'}{\typ}
  {\smtappname{\embed{\typ_x}}{\embed{\typ}}\ {\pred}\ {\pred'}}{\embed{\typ}}{\smtenv}{\andaxioms{\axioms}{\axioms'}}. 
By inversion, 
	\tologicshort{\env}{e'}{\typ_x}{\pred'}{\embed{\typ_x}}{\smtenv}{\axioms'}, 
	\tologicshort{\env}{e}{\tfun{x}{\typ_x}{\typ}}{\pred}{\tsmtfun{\embed{\typ_x}}{\embed{\typ}}}{\smtenv}{\axioms},
	\hastype{\env}{e}{\tfun{x}{\typ_x}{\typ}}. 
By IH and the type of $\smtappname{}{}$ we get that 
   \smthastype{\smtenvinit, \embed{\env}}
     {\smtappname{\embed{\typ_x}}{\embed{\typ}}\ {\pred}\ {\pred'}}
     {\embed{\typ}}.

\item \lgcaseBool : 
Assume
	\tologicshort{\env}{\ecaseexp{x}{e}{\etrue \rightarrow e_1; \efalse \rightarrow e_2}}{\typ}
	 {\eif{\pred}{\pred_1}{\pred_2}}{\embed{\typ}}{\smtenv}{\andaxioms{\axioms}{\axioms_i}}
Since 
  \hastype{\env}{\ecaseexp{x}{e}{\etrue \rightarrow e_1; \efalse \rightarrow e_2}}{\typ}, then 
  \hastype{\env}{e}{\tbool}, 
  \hastype{\env}{e_1}{\typ}, and
  \hastype{\env}{e_2}{\typ}.
By inversion and IH, 
  \smthastype{\smtenvinit, \embed{\env}}{\pred}{\tbool}, 
  \smthastype{\smtenvinit, \embed{\env}}{\pred_1}{\embed{\typ}}, and
  \smthastype{\smtenvinit, \embed{\env}}{\pred_2}{\embed{\typ}}.
Thus, 
  \smthastype{\smtenvinit, \embed{\env}}{\eif{\pred}{\pred_1}{\pred_2}}{\embed{\typ}}.

\item \lgcase : 
Assume 
	\tologicshort{\env}{\ecase{x}{e}{\dc_i}{\overline{y_i}}{e_i}}{\typ}
	 {\eif{\checkdc{\dc_1}\ \pred}{\pred_1}{\ldots} \ \mathtt{else}\ \pred_n}{\embed{\typ}}{\smtenv}
	 {\andaxioms{\axioms}{\axioms_i}} and
	 \hastype{\env}{\ecase{x}{e}{\dc_i}{\overline{y_i}}{e_i}}{\typ}.
By inversion we get 
	\tologicshort{\env}{e}{\typ_e}{\pred}{\embed{\typ_e}}{\smtenv}{\axioms} and
	\tologicshort{\env}{e_i\subst{\overline{y_i}}{\overline{\selector{\dc_i}{}\ x}}\subst{x}{e}}{\typ}{\pred_i}{\embed{\typ}}{\smtenv}{\axioms_i}.
By IH and the Definition~\ref{def:initialsmt} on the checkers and selectors, we get
  \smthastype{\smtenvinit, \embed{\env}}{\eif{\checkdc{\dc_1}\ \pred}{\pred_1}{\ldots} \ \mathtt{else}\ \pred_n}{\embed{\typ}}.
\end{itemize}
\end{proof}

\newcommand\tsub[1]{\ensuremath{{\theta_{#1}^\perp}}\xspace}
\newcommand\track[2]{\ensuremath{\langle #1; #2\rangle}\xspace}
\begin{theorem}\label{thm:approximation}
If \tologicshort{\env}{\refa}{}{\pred}{}{}{},
then for every substitution $\sub\in\interp{\env}$
and every model $\sigma\in\interp{\theta^\perp}$,
if $\evalsto{\applysub{\theta^\perp}{\refa}}{v}$ then $\sigma^\beta \models \pred = \interp{v}$.
\end{theorem}
\begin{proof}
We proceed using the notion of tracking substitutions from Figure 8 of~\citep{Vazou14-tech}. 
Since $\evalsto{\applysub{\theta^\perp}{\refa}}{v}$, 
there exists a sequence of evaluations via tracked substitutions, 
$$
\track{\tsub{1}}{e_1} \hookrightarrow \dots \track{\tsub{i}}{e_i} \dots \hookrightarrow \track{\tsub{n}}{e_n}
$$
with $\tsub{1}\equiv\tsub{}$, $e_1\equiv e$, and $e_n\equiv v$. 
Moreover, each $e_{i+1}$ is well formed under $\Gamma$, 
thus it has a translation 
$\tologicshort{\Gamma}{\refa_{i+1}}{}{\pred_{i+1}}{}{}{}$. 
Thus we can iteratively apply Lemma~\ref{lemma:approximation} $n-1$ times and 
since $v$ is a value the extra variables in $\tsub{n}$ are irrelevant, thus we
get the required
$\sigma^\beta \models \pred = \interp{v}$. 
\end{proof}

For Boolean expressions we specialize the above to
\begin{corollary}\label{thm:embedding}
If \hastype{\env}{\refa}{\tbool^\downarrow} and
\tologicshort{\env}{\refa}{\tbool}{\pred}{\tbool}{\smtenv}{\axioms},
then for every substitution $\sub\in\interp{\env}$
and every model $\sigma\in\interp{\theta^\perp}$,
$\evalsto{\applysub{\theta^\perp}{\refa}}{\etrue} \iff \sigma^\beta \models \pred $
\end{corollary}
\begin{proof}
We prove the left and right implication separately:
\begin{itemize}
\item $\Rightarrow$
By direct application of Theorem~\ref{thm:approximation} for $v \equiv \etrue$. 

\item $\Leftarrow$ 
Since $\refa$ is terminating, 
$\evalsto{\applysub{\theta^\perp}{\refa}}{v}$. 
with either $v \equiv \etrue$ or $v \equiv \efalse$. 
Assume $v \equiv \efalse$, then by Theorem~\ref{thm:approximation}, 
$\bmodel \models \lnot \pred$, which is a contradiction. 
Thus, $v \equiv \etrue$.
\end{itemize}
\end{proof}

\begin{lemma}[Equivalence Preservation]\label{lemma:approximation}
If \tologicshort{\env}{\refa}{}{\pred}{}{}{},
then for every substitution $\sub\in\interp{\env}$
and every model $\sigma\in\interp{\theta^\perp}$,
if  
$\track{\tsub{}}{\refa}\hookrightarrow\track{\tsub{2}}{\refa_2}$
and
for $\Gamma \subseteq \Gamma_2$ so that $\tsub{2} \in \interp{\Gamma_2}$
and $\bmodel_2\in\interp{\tsub{2}}$,
$\tologicshort{\Gamma_2}{\refa_2}{}{\pred_2}{}{}{}$
then 
$\bmodel \cup (\bmodel_2 \setminus \bmodel) \models \pred = \pred_2$.
\end{lemma}

\begin{proof}
We proceed by case analysis on the derivation 
$\track{\tsub{}}{\refa}\hookrightarrow\track{\tsub{2}}{\refa_2}$.
\begin{itemize}
\item 
Assume
	$\track{\tsub{}}{\refa_1\ \refa_2}\hookrightarrow\track{\tsub{2}}{\refa_1'\ \refa_2}$. 
By inversion 
	$\track{\tsub{}}{\refa_1}\hookrightarrow\track{\tsub{2}}{\refa_1'}$.
Assume
  $\tologicshort{\Gamma}{\refa_1}{}{\pred_1}{}{}{}$, 	
  $\tologicshort{\Gamma}{\refa_2}{}{\pred_2}{}{}{}$, 	
  $\tologicshort{\Gamma_2}{\refa_1'}{}{\pred_1'}{}{}{}$.  	
By IH 
  $\bmodel \cup (\bmodel_2 \setminus \bmodel)
    \models \pred_1 = \pred_1'$, 
thus  
  $\bmodel \cup (\bmodel_2 \setminus \bmodel)
    \models \smapp{\pred_1}{\pred_2} = \smapp{\pred_1'}{\pred_2}$. 
    
\item 
Assume
	$\track{\tsub{}}{c\ \refa}\hookrightarrow\track{\tsub{2}}{c\ \refa'}$. 
By inversion 
	$\track{\tsub{}}{\refa}\hookrightarrow\track{\tsub{2}}{\refa'}$.
Assume
  $\tologicshort{\Gamma}{\refa}{}{\pred}{}{}{}$,	
  $\tologicshort{\Gamma}{\refa'}{}{\pred'}{}{}{}$.  	
By IH 
  $\bmodel \cup (\bmodel_2 \setminus \bmodel)
    \models \pred = \pred'$, 
thus  
  $\bmodel \cup (\bmodel_2 \setminus \bmodel)
    \models \smapp{c}{\pred} = \smapp{c}{\pred'}$. 

\item 
Assume
	$\track{\tsub{}}{\ecase{x}{\refa}{\dc_i}{\overline{y_i}}{e_i}}
	  \hookrightarrow\track
	  {\tsub{2}}{\ecase{x}{\refa'}{\dc_i}{\overline{y_i}}{e_i}}$. 
By inversion 
	$\track{\tsub{}}{\refa}\hookrightarrow\track{\tsub{2}}{\refa'}$.
Assume
  $\tologicshort{\Gamma}{\refa}{}{\pred}{}{}{}$,	
  $\tologicshort{\Gamma}{\refa'}{}{\pred'}{}{}{}$.  	
  $\tologicshort{\env}{e_i\subst{\overline{y_i}}{\overline{\selector{\dc_i}{}\ x}}\subst{x}{e}}{\typ}{\pred_i}{\embed{\typ}}{\smtenv}{\axioms_i}$.
By IH 
  $\bmodel \cup (\bmodel_2 \setminus \bmodel)
    \models \pred = \pred'$, 
thus  
  $\bmodel \cup (\bmodel_2 \setminus \bmodel)
   \models
	 {\eif{\checkdc{\dc_1}\ \pred}{\pred_1}{\ldots} \ \mathtt{else}\ \pred_n}{\embed{\typ}}
 $\\ $=
	 {\eif{\checkdc{\dc_1}\ \pred'}{\pred_1}{\ldots} \ \mathtt{else}\ \pred_n}{\embed{\typ}}
  $.	 

\item 
Assume
	$\track{\tsub{}}{D\ \overline{\refa_i}\ \refa\ \overline{\refa_j}}
	\hookrightarrow\track{\tsub{2}}{D\ \overline{\refa_i}\ \refa'\ \overline{\refa_j}}$. 
By inversion 
	$\track{\tsub{}}{\refa}\hookrightarrow\track{\tsub{2}}{\refa'}$.
Assume
  $\tologicshort{\Gamma}{\refa}{}{\pred}{}{}{}$,	
  $\tologicshort{\Gamma}{\refa_i}{}{\pred_i}{}{}{}$,	
  $\tologicshort{\Gamma}{\refa'}{}{\pred'}{}{}{}$.  	
By IH 
  $\bmodel \cup (\bmodel_2 \setminus \bmodel)
    \models \pred = \pred'$, 
thus  
  $\bmodel \cup (\bmodel_2 \setminus \bmodel)
    \models \smapp{D}{\overline{\pred_i}\ \pred\ \overline{\pred_j}} 
          = \smapp{D}{\overline{\pred_i}\ \pred'\ \overline{\pred_j}}$. 

\item 
Assume
	$\track{\tsub{}}{c\ w}
	\hookrightarrow\track{\tsub{}}{\delta(c, w)}$. 
By the definition of the syntax, $c\ w$ is a fully applied logical operator, thus
  $\bmodel \cup (\bmodel_2 \setminus \bmodel)
    \models {c\ w} = \interp{\delta(c, w)}$

\item 
Assume
	$\track{\tsub{}}{(\efun{x}{}{\refa}) \refa_x}
	\hookrightarrow\track{\tsub{}}{\refa\subst{x}{\refa_x}}$. 
Assume
  $\tologicshort{\Gamma}{\refa}{}{\pred}{}{}{}$,	
  $\tologicshort{\Gamma}{\refa_x}{}{\pred_x}{}{}{}$. 
Since $\bmodel$ is defined to satisfy the $\beta$-reduction axiom, 
  $\bmodel \cup (\bmodel_2 \setminus \bmodel)
    \models \smapp{(\smlam{x}{\refa})}{\pred_x} =\pred\subst{x}{\pred_x}$. 
  
\item 
Assume
	$\track{\tsub{}}{\ecase{x}{D_j\ \overline{e}}{\dc_i}{\overline{y_i}}{e_i}}
	\hookrightarrow\track{\tsub{}}{e_j\subst{x}{D_j\ \overline{e}}\subst{y_i}{\overline{e}}}$. 
Also, let 
  $\tologicshort{\Gamma}{\refa}{}{\pred}{}{}{}$,	
  $\tologicshort{\Gamma}{\refa_i\subst{x}{D_j\ \overline{\refa}}\subst{y_i}{\overline{\refa_i}}}{}{\pred_i}{}{}{}$.
By the axiomatic behavior of the measure selector $\checkdc{\dc_j\ \overline{\pred}}$, we get 
  $\bmodel\models \checkdc{\dc_j\ \overline{\pred}}$.
Thus, 
  $\bmodel
	 {\eif{\checkdc{\dc_1}\ \pred}{\pred_1}{\ldots} \ \mathtt{else}\ \pred_n}
	= \pred_j
  $.

\item 
Assume
	$\track{(x, e_x)\tsub{}}{x}
	\hookrightarrow
	\track{(x, e_x')\tsub{2}}{x}$.
By inversion
	$\track{\tsub{}}{e_x}
	\hookrightarrow
	\track{\tsub{2}}{e_x'}$.
By identity of equality, 
  $(x, \pred_x)\bmodel \cup (\bmodel_2 \setminus \bmodel)
    \models x = x$. 

\item 
Assume
	$\track{(y, e_y)\tsub{}}{x}
	\hookrightarrow
	\track{(y, e_y)\tsub{2}}{e_x}$.
By inversion
	$\track{\tsub{}}{x}
	\hookrightarrow
	\track{\tsub{2}}{e_x}$.
Assume 
  $\tologicshort{\Gamma}{\refa_x}{}{\pred_x}{}{}{}$.
By IH
  $\bmodel \cup (\bmodel_2 \setminus \bmodel)
    \models x = \pred_x$. 
Thus  
  $(y,\pred_y)\bmodel \cup (\bmodel_2 \setminus \bmodel)
    \models x = \pred_x$. 

\item 
Assume
	$\track{(x, w)\tsub{}}{x}
	\hookrightarrow
	\track{(x, w)\tsub{}}{w}$.
Thus  
  $(x,\interp{w})\bmodel 
    \models x = \interp{w}$. 

\item 
Assume
	$\track{(x, D\ \overline{y})\tsub{}}{x}
	\hookrightarrow
	\track{(x, D\ \overline{y})\tsub{}}{D\ \overline{y}}$.
Thus  
  $(x,\smapp{D}{\overline{y}})\bmodel 
    \models x = \smapp{D}{\overline{y}}$. 

\item 
Assume
	$\track{(x, D\ \overline{e})\tsub{}}{x}
	\hookrightarrow
	\track{(x, D\ \overline{y}), \overline{(y_i, e_i)}\tsub{}}{D\ \overline{y}}$.
Assume 
  $\tologicshort{\Gamma}{\refa_i}{}{\pred_i}{}{}{}$.
Thus  
  $(x,\smapp{D}{\overline{y}}), \overline{(y_i, \pred_i)}\bmodel 
    \models x = \smapp{D}{\overline{y}}$. 
\end{itemize}
\end{proof}

\subsection{Soundness of Approximation}
\begin{theorem}[Soundness of Algorithmic]
If \ahastype{\env}{e}{\typ} then \hastype{\env}{e}{\typ}.
\end{theorem}
\begin{proof}
To prove soundness it suffices to prove that subtyping is appropriately approximated, 
as stated by the following lemma.
\end{proof}

\begin{lemma}
If \aissubtype{\env}{\tref{v}{\btyp}{e_1}}{\tref{v}{\btyp}{e_2}}
then \issubtype{\env}{\tref{v}{\btyp}{e_1}}{\tref{v}{\btyp}{e_2}}.
\end{lemma}
\begin{proof}
By rule \rsubbase, we need to show that
$\forall \sub\in\interp{\env}.
  \interp{\applysub{\sub}{\tref{v}{\btyp}{\refa_1}}}
  \subseteq
  \interp{\applysub{\sub}{\tref{v}{\btyp}{\refa_2}}}$.
We fix a $\sub\in\interp{\env}$.
and get that forall bindings
$(\tbind{x_i}{\tref{v}{\btyp^{\downarrow}}{\refa_i}}) \in \env$,
$\evalsto{\applysub{\sub}{e_i\subst{v}{x_i}}}{\etrue}$.

Then need to show that for each $e$,
if $e \in \interp{\applysub{\sub}{\tref{v}{\btyp}{\refa_1}}}$,
then $e \in \interp{\applysub{\sub}{\tref{v}{\btyp}{\refa_2}}}$.

If $e$ diverges then the statement trivially holds.
Assume $\evalsto{e}{w}$.
We need to show that
if $\evalsto{\applysub{\sub}{e_1\subst{v}{w}}}{\etrue}$
then $\evalsto{\applysub{\sub}{e_2\subst{v}{w}}}{\etrue}$.

Let \vsub the lifted substitution that satisfies the above.
Then  by Lemma~\ref{thm:embedding}
for each model $\bmodel \in \interp{\vsub}$,
$\bmodel\models\pred_i$, and $\bmodel\models q_1$
for
$\tologicshort{\env}{e_i\subst{v}{x_i}}{\btyp}{\pred_i}{\embed{\btyp}}{\smtenv_i}{\axioms_i}$
$\tologicshort{\env}{e_i\subst{v}{w}}{\btyp}{q_i}{\embed{\btyp}}{\smtenv_i}{\beta_i}$.
Since \aissubtype{\env}{\tref{v}{\btyp}{e_1}}{\tref{v}{\btyp}{e_2}} we get
$$
\bigwedge_i \pred_i
\Rightarrow q_1 \Rightarrow q_2
$$
thus $\bmodel\models q_2$.
By Theorem~\ref{thm:embedding} we get $\evalsto{\applysub{\sub}{\refa_2\subst{v}{w}}}{\etrue}$.
\end{proof}

}

\end{document}